\documentclass[12pt]{article}
\usepackage{amsfonts}

\usepackage{mathrsfs}
\usepackage{srcltx}
\usepackage{xcolor}
\usepackage{amsmath}
\usepackage{amsthm}
\usepackage{amstext}
\usepackage{amsopn}

\usepackage{subfigure}
\usepackage{tikz}

\usepackage{graphicx}
\usepackage{bm}
\newtheorem{theorem}{Theorem}[section]
\newtheorem{lemma}[theorem]{Lemma}

\theoremstyle{definition}

\theoremstyle{remark}
\newtheorem{remark}[theorem]{Remark}

\numberwithin{equation}{section}

\newcommand{\ba}{\begin{array}}
\newcommand{\ea}{\end{array}}
\newcommand{\f}{\frac}

\newcommand{\la}{\lambda}

\newcommand{\ds}{\displaystyle}

\begin{document}
\date{}
\title{ \bf\large{Evolution of dispersal in advective patchy environments with varying drift rates}\thanks{S. Chen is supported by National Natural Science Foundation of China (Nos. 12171117) and Shandong Provincial Natural Science Foundation of China (No. ZR2020YQ01). 
}}
\author{
Shanshan Chen\footnote{Email: chenss@hit.edu.cn}\\[-1mm]
{\small Department of Mathematics, Harbin Institute of Technology}\\[-2mm]
{\small Weihai, Shandong 264209, P. R. China}\\[1mm]
Jie Liu\footnote{Corresponding Author, Email: liujie9703@126.com}\\[-1mm]
{\small School of Mathematics, Harbin Institute of Technology}\\[-2mm]
{\small Harbin, Heilongjiang, 150001, P. R. China}\\[1mm]
Yixiang Wu\footnote{Email: yixiang.wu@mtsu.edu} \\[-1mm]
{\small Department of Mathematics, Middle Tennessee State University}\\[-2mm]
{\small Murfreesboro, Tennessee 37132, USA}\\[1mm]
}

\maketitle
\begin{abstract}
In this paper, we study a two stream species Lotka-Volterra competition patch model with the patches aligned along a line. The two species are supposed to be identical except for the diffusion rates. For each species, the diffusion rates between patches are the same, while the drift rates vary. 
Our results show that the convexity of the drift rates has a significant impact on the competition outcomes:
if the drift rates are convex, then the species with larger diffusion rate wins the competition; if the drift rates are concave, then the species with smaller diffusion rate wins the competition. \\[2mm]
\noindent {\bf Keywords}: Lotka-Volterra competition model, patch, evolution of dispersal. \\[2mm]
\noindent {\bf MSC 2020}: 92D25, 92D40, 34C12, 34D23, 37C65.
\end{abstract}

\section{Introduction}
In stream ecology, one puzzling question named ``drift paradox'' asks why aquatic species in advective stream environments can persist \cite{pachepsky2005persistence}. In an attempt to answer this question, Speirs and Gurney propose a single species reaction-diffusion-advection model with Fisher-KPP type nonlinear term and show that the diffusive dispersal can permit persistence in advective environment \cite{speirs2001population}. This work has inspired a series of studies on how factors such as  diffusion rate, advection rate, domain size, and spatial heterogeneity impact the persistence of a stream species \cite{Huang16,Jin11,lam2016emergence,Lou14, Lou15, Lutscher06,Lutscher05,zhou2018evolution,zhou2018global}. 

One natural extension of the work \cite{speirs2001population} is to consider  the competition of two species in stream environment, where both spices are subjective to random dispersal and passive directed drift (see \cite{Lou14, Lou15, Lutscher06,Lutscher05,vasilyeva2012flow,wang2019persistence, yan2022competition,MR4222368,zhou2018evolution,zhou2018global} and the references therein). One interesting research direction for competition models is to study the evolution of dispersal. 
Earlier results in \cite{dockery1998evolution,hastings1983can} claim that the species with a slower  movement rate has competitive advantage in a spatially heterogeneous environment when both competing species have only random dispersal patterns. Later, it is shown that faster diffuser can be selected in an advective  environment  (e.g., see \cite{MR2332679,MR2462700, Cantrell12, Chenx12}). For reaction-diffusion-advection competition models of stream species, Lou and Lutscher \cite{Lou14}   show that the species with a larger diffusion rate may have competitive advantage when the two competing species are only different by the diffusion rates; Lam \emph{et al.} \cite{Lam15} prove that an intermediate dispersal rate may be selected if the resource function is spatial dependent; Lou \emph{et al.} \cite{Lou19} show that the species with larger dispersal rate has competitive advantage when the resource function is decreasing and the drift rate is large; If the resource function is increasing, the results of Tang \emph{et al.} \cite{Tang2022} indicate that the species with slower diffusion rate may prevail. It seems like that the role of spatial heterogeneity in advection rate is less studied. To our best knowledge, the only such work on reaction-diffusion-advection competition models for stream species  is  by Shao \emph{et al.}  \cite{shao2022second}, which shows that the slower diffuser may win if the advection rate function is concave. 





Our study is also motivated by a series of works on competition patch models of the form \cite{Chen22,Cheng19,Cheng21,Hamida2017,Jiang20,Jiang21, lou2019ideal, noble2015evolution, Xiang19}:
\begin{equation}\label{pat-cp}
\begin{cases}
\ds\frac{du_i}{dt}=\ds\sum_{j=1}^{n} \left(d_1 D_{ij}+q_1Q_{ij}\right)u_j+u_i(r-u_i-v_i), &i=1,\cdots,n,\;\;t>0,\\
\ds\frac{dv_i}{dt}=\ds\sum_{j=1}^{n}\left(d_2 D_{ij}+q_2Q_{ij}\right)v_j+v_i(r-u_i-v_i),&i=1,\dots,n,\;\; t>0, \\
\bm u(0)=\bm u_0\ge(\not\equiv)\bm0,\;\bm v(0)=\bm v_0\ge(\not\equiv)\bm0,
\end{cases}
\end{equation}
where $\bm u=(u_1,\dots, u_n)$ and $\bm v=(v_1,\dots, v_n)$ are population density of two competing species living in stream environment; matrices $D$ and $Q$ describe the diffusion and  drift  patterns, respectively; $d_1, d_2$ are the diffusion rates and $q_1, q_2$ are the advection rates. Recent works on model \eqref{pat-cp} for small number of patches can be found in the literature (see \cite{Cheng19,Cheng21,Hamida2017, lou2019ideal, noble2015evolution, Xiang19} for $n=2$ and  \cite{chen2022,Jiang20} for $n=3$). In particular, Jiang  \emph{et al.} \cite{Jiang20,Jiang21} propose  three configurations of three-node stream networks (i.e. $n=3$) and show that the magnitude of drift rate can affect whether the slower or faster diffuser wins the competition. Chen \emph{et al.} \cite{chen2022, chen2022evolution} generalize the   configuration in \cite{Jiang20} when all the nodes were aligned alone a line  to arbitrarily many nodes and have studied \eqref{pat-cp} for three stream networks that are different only in the downstream end.   



In this paper, we will consider the following two stream species competition patch model:
\begin{equation}\label{pat-cp2}
\begin{cases}
\ds\frac{du_i}{dt}=\ds\sum_{j=1}^{n} \left(d_1 D_{ij}+q_jQ_{ij}\right)u_j+u_i(r-u_i-v_i), &i=1,\cdots,n,\;\;t>0,\\
\ds\frac{dv_i}{dt}=\ds\sum_{j=1}^{n}\left(d_2 D_{ij}+q_jQ_{ij}\right)v_j+v_i(r-u_i-v_i),&i=1,\dots,n,\;\; t>0, \\
\bm u(0)=\bm u_0\ge(\not\equiv)\bm0,\;\bm v(0)=\bm v_0\ge(\not\equiv)\bm0,
\end{cases}
\end{equation}
where $u_i$ and $v_i$ denote the population densities of the two competing species in patch $i$ at time $t$, respectively; $r>0$ is the intrinsic growth rate and represents the effect of resources; $d_1,d_2>0$ represent random movement rates; $q_i\ge0$, $i=1,\cdots,n$, are directed movement rates. The patches are aligned a long a line as shown in Fig. \ref{network} (let $n$ be the number of patches, and we always suppose $n\ge 2$). 
\begin{figure}[htbp]
\label{network}
\centering
\begin{tikzpicture}
\begin{scope}[every node/.style={draw}, node distance= 1.5 cm]
    \node[circle] (1) at (0,0) {$1$};
    \node[circle] (2) at (2,0) {$2$};
    \node[circle] (3) at (4,0) {$3$};
    \node[circle] (4) at (6,0) {$4$};

    \node[circle] (5) at (8,0) {$5$};
    \node[circle] (6) at (10,0) {$6$};
\end{scope}
\begin{scope}[every node/.style={fill=white},
              every edge/.style={thick}]
    \draw[thick] [->](1) to [bend right] node[below=0.1] {{\footnotesize $d+q_1$}} (2);
    \draw[thick] [->](2) to [bend right] node[below=0.1] {{\footnotesize $d+q_2$}} (3);
    \draw[thick] [->](3) to [bend right] node[below=0.1] {{\footnotesize $d+q_3$}} (4);
    \draw[thick] [->](4) to [bend right] node[below=0.1] {{\footnotesize $d+q_4$}} (5);
    \draw[thick] [->](5) to [bend right] node[below=0.1] {{\footnotesize $d+q_5$}} (6);
    \draw[thick] [<-](1) to [bend left] node[above=0.1] {{\footnotesize $d$}} (2);
    \draw[thick] [<-](2) to [bend left] node[above=0.1] {{\footnotesize $d$}} (3);
    \draw[thick] [<-](3) to [bend left] node[above=0.1] {{\footnotesize $d$}} (4);
    \draw[thick] [<-](4) to [bend left] node[above=0.1] {{\footnotesize $d$}} (5);
    \draw[thick] [<-](5) to [bend left] node[above=0.1] {{\footnotesize $d$}} (6);
    \draw[thick] [->](6) to  node[above=0.1] {{\footnotesize $q_6$}} (12, 0);
\end{scope}
\end{tikzpicture}
\caption{A stream with six patches, where $d$ is the random movement rate, and $q_i$, i=1,\ldots,6, are the directed movement rates. }
\end{figure}
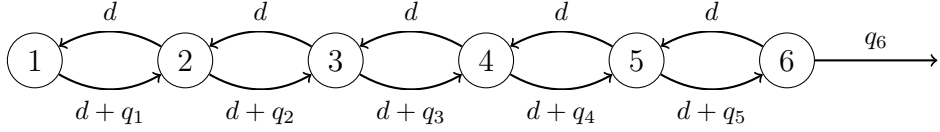
If  $q_n=0$, then the stream is an inland stream, and there exists no population loss at the downstream end; if $q_n>0$, it corresponds to the situation that the stream flows to a lake, where the diffusive flux into and from the lake balances. 
The $n\times n$ matrices $D=(D_{ij})$ and $Q=(Q_{ij})$ describe the diffusion and drift patterns, respectively, where
\begin{equation}\label{M}
\begin{split}
D_{ij}=&\begin{cases}
1,&i=j-1\;\text{or}\;i=j+1,\\
-2,&i=j=2,\cdots,n-1,\\
-1,&i=j=1,n,\\
0, &\text{otherwise},
\end{cases}
\end{split}\;\;\begin{split}
Q_{ij}=&\begin{cases}
1,&i=j+1,\\
-1,&i=j=1,\cdots,n,\\
0, &\text{otherwise}.
\end{cases}
\end{split}
\end{equation}
Our work is motivated by Jiang \emph{et al.} \cite{Jiang21}, and we consider model \eqref{pat-cp2} when $n\ge3$. We will show that the convexity of the drift rate $\bm q=(q_1, \dots, q_n)$ affects the evolution of random dispersal. In particular, if $\bm q=(q_1, \dots, q_n)$ is convex, then the species with a larger diffusion rate has competitive advantage; if $\bm q$ is concave, then the slower diffusier has competitive advantage. 


Our paper is organized as follows. In section 2, we state the main results on the global dynamics of model \eqref{pat-cp2}.  The other sections  are about the details of the proofs for the main results: in section 3, we consider the eigenvalue problems that are related to the existence and stability of the semi-trivial equilibria; in section 4, we study the existence and non-existence of semi-trivial equilibria; an essential step to prove the competitive exclusion result for the model is to show that the model has no positive equilibrium, which is placed in section 5; in section 6, we study the local stability of the semi-trivial equilibria.   


\section{Main results}
In this section, we state our main results about the global dynamics of model \eqref{pat-cp2}. Let $q_0=0$ throughout the paper. We will consider  \eqref{pat-cp2} under two different scenarios for the flow rate $\bm q=(q_1, \dots, q_n)\ge\bm 0$:
\begin{enumerate}
\item [$(\textbf{H1})$] $q_{i+1}-q_{i}\le q_{i}-q_{i-1}$ for $i=1,\cdots,n-1$, with at least one strict inequality.
\item [$(\textbf{H2})$] $q_{i+1}-q_{i}\ge q_{i}-q_{i-1}$ for $i=1,\cdots,n-1$, with at least one strict inequality.
\end{enumerate}

Fig. \ref{fig2} illustrates the flow rate $\bm q$ for the model with six patches under $(\textbf{H1})$ and $(\textbf{H2})$. Assumption $(\textbf{H1})$  describes a stream whose flow rate is convex, while $(\textbf{H2})$ depicts the situation that the flow rate is concave. We will show that the competition outcomes of the species are dramatically different under  $(\textbf{H1})$ or $(\textbf{H2})$. We remark that our assumption $(\textbf{H2})$ implicitly implies that the components of $\bm q$ are strictly increasing as  $q_0=0$ and $\bm q\ge \bm 0$.

\begin{figure}[htbp]
\centering\includegraphics[width=0.45\textwidth]{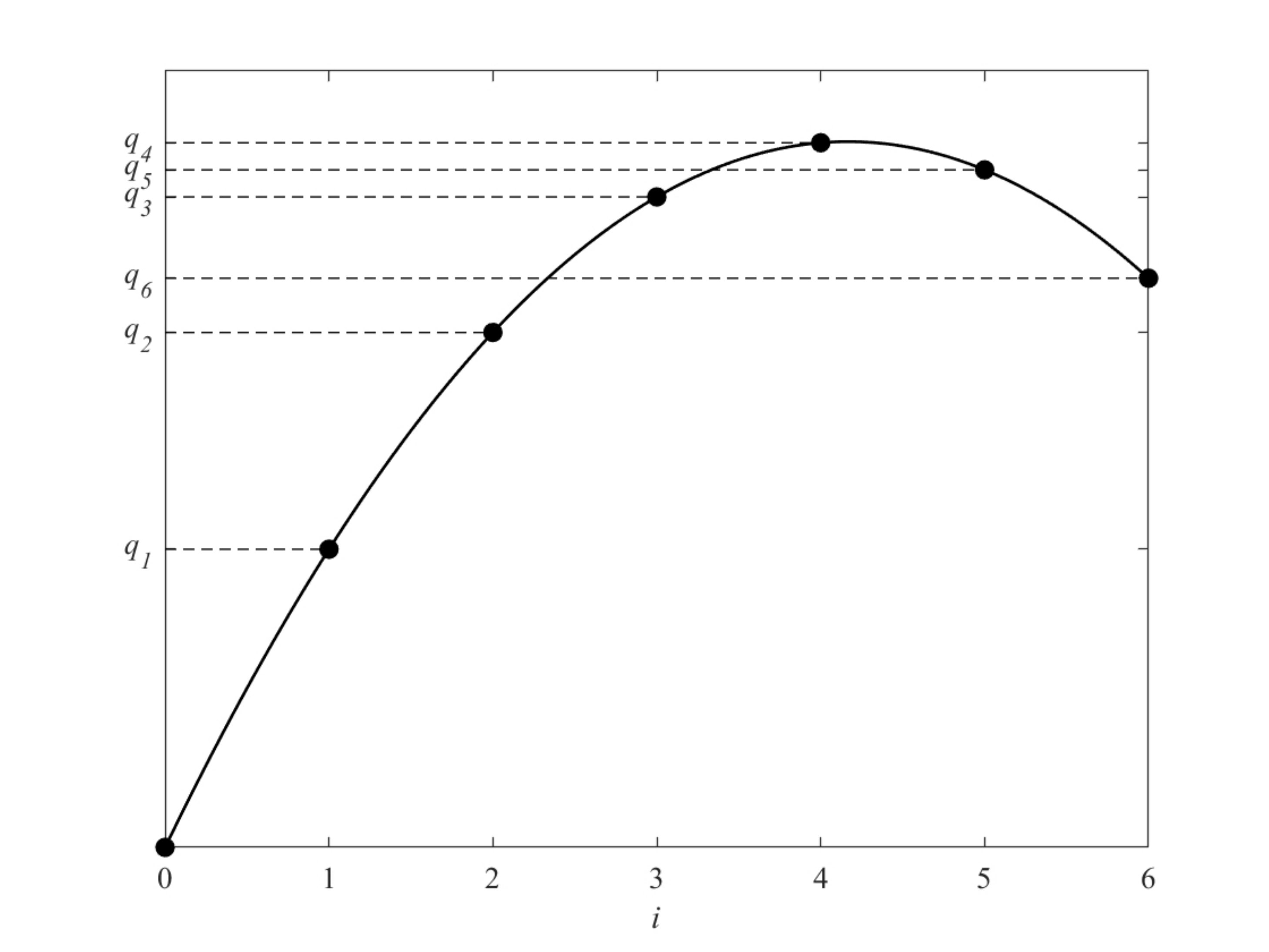}
\centering\includegraphics[width=0.45\textwidth]{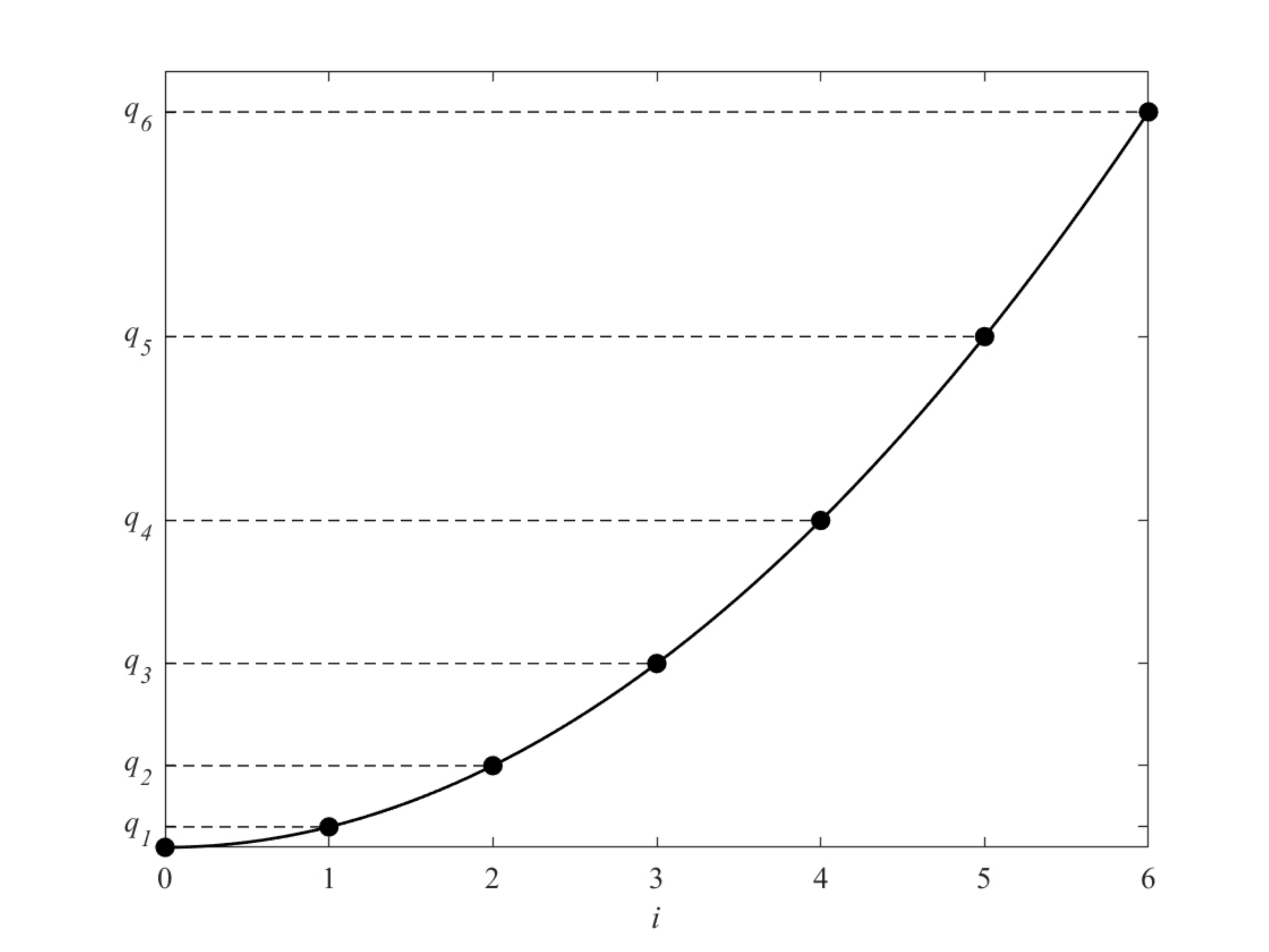}
\caption{Illustration of the drift rate $\bm q$ for model \eqref{pat-cp2} with six patches. The left figure satisfies $(\textbf{H1})$ and the right figure satisfies $(\textbf{H2})$.}
\label{fig2}
\end{figure}

First, we consider the case  that the drift rate $\bm q$ satisfies assumption $(\rm{\bf{H1}})$.
\begin{theorem}\label{dl}
Suppose that $d_1, d_2>0$, $\bm q\ge \bm 0$, and $(\rm{\bf{H1}})$ holds. Then  the following statements on model \eqref{pat-cp2} hold:
\begin{enumerate}
\item [${\rm (i)}$]  If $\ds\min_{1\le i\le n}q_i<r$, then the model has two semi-trivial equilibria $(\bm u^*, \bm 0)$ and $(\bm 0,\bm v^*)$. Moreover,
  \begin{itemize}
    \item[${\rm (i_{1})}$] if $d_1<d_2$, then $(\bm 0,\bm v^*)$ is globally asymptotically stable;
    \item[${\rm (i_{2})}$] if $d_1>d_2$, then $(\bm u^*,\bm 0)$ is globally asymptotically stable;
  \end{itemize}
\item [${\rm (ii)}$] If $\ds\min_{1\le i\le n}q_i>r$ and $q_n<rn$, then there exists $d^*>0$ (obtained in Lemma \ref{dzq}) such that
    \begin{enumerate}
    \item [${\rm (ii_{1})}$] if $d_2>d^*$ and $d_2>d_1$, then the  semi-trivial equilibrium $(\bm 0,\bm v^*)$ exists, which is globally asymptotically stable;
    \item [${\rm (ii_{2})}$] if $d_1>d^*$ and $d_1>d_2$, then the semi-trivial equilibrium $(\bm u^*,\bm 0)$ exists, which is globally asymptotically stable;
    \item [${\rm (ii_{3})}$] if $0<d_1,d_2\le d^*$, then the trivial equilibrium $(\bm 0, \bm 0)$ is globally asymptotically stable;
    \end{enumerate}
\item [${\rm (iii)}$] If $\ds\min_{1\le i\le n}q_i>r$ and $q_n>rn$, then the trivial equilibrium $(\bm 0, \bm 0)$ is globally asymptotically stable.
\end{enumerate}
\end{theorem}

\begin{remark}
If $n=3$, similar results obtained in \cite{Jiang21} require weaker assumptions than $(\rm{\bf{H1}})$. It seems  that the methods in \cite{Jiang21} cannot be generalized to the case $n>3$.
\end{remark}

Then, we consider the case  that the drift rate $\bm q$ satisfies assumption $(\rm{\bf{H2}})$.

\begin{theorem}\label{dl0}
Suppose that $d_1, d_2>0$, $\bm q\ge \bm 0$, and $(\rm{\bf{H2}})$ holds. Then the the following statements on model \eqref{pat-cp2} hold:
\begin{enumerate}
\item [${\rm (i)}$]  If $q_n<rn$, then both semi-trivial equilibria $(\bm u^*, \bm 0)$ and $(\bm 0,\bm v^*)$  exist. Moreover,
  \begin{enumerate}
    \item[${\rm (i_{1})}$] if $d_1>d_2$, then $(\bm 0,\bm v^*)$ is globally asymptotically stable;
    \item[${\rm (i_{2})}$] if $d_1<d_2$, then $(\bm u^*,\bm 0)$ is globally asymptotically stable;
  \end{enumerate}
\item [${\rm (ii)}$] if $q_n>rn$ and $q_1<r$, then there exists $\widetilde d^*>0$ (obtained in Lemma \ref{dzq2}) such that
    \begin{enumerate}
    \item [${\rm (ii_{1})}$] if $d_2<\widetilde d^*$ and $d_2<d_1$, then the semi-trivial equilibrium $(\bm 0,\bm v^*)$ exists and is globally asymptotically stable;
    \item [${\rm (ii_{2})}$] if $d_1<\widetilde d^*$ and $d_1<d_2$, then the semi-trivial equilibrium $(\bm u^*,\bm 0)$ exists and is globally asymptotically stable;
    \item [${\rm (ii_{3})}$] if $d_1,d_2\ge\widetilde d^*$, then the trivial equilibrium $\bm 0$ is globally asymptotically stable;
    \end{enumerate}
\item [${\rm (iii)}$] if $q_n>rn$ and $q_1>r$, then the trivial equilibrium $\bm 0$ is globally asymptotically stable.
\end{enumerate}
\end{theorem}

\begin{proof}[Proof of Theorems \ref{dl} and \ref{dl0}] 


It is well-known that the competition system \eqref{pat-cp2} induces a monotone dynamical system \cite{smith2008monotone}. By the monotone dynamical system theory, 
if both semi-trivial equilibria exist,  $(\bm u^*, \bm 0)$  (resp. $(\bm 0, \bm v^*)$) is unstable and the model has no positive equilibrium, then $(\bm 0, \bm v^*)$ (resp. $(\bm u^*, \bm 0)$) is globally asymptotically stable \cite{Hess,hsu1996competitive, Lam-Munther2016,smith2008monotone}. Moreover by Lemma \ref{DS-single},
if both semi-trivial equilibria do not exist, then the trivial equilibrium $(\bm 0, \bm 0)$ is globally asymptotically stable; if $(\bm u^*, \bm 0)$ (resp. $(\bm 0, \bm v^*)$) exists and $(\bm 0, \bm v^*)$ (resp. $(\bm u^*, \bm 0)$) does not exist, then $(\bm u^*, \bm 0)$ (resp. $(\bm 0, \bm v^*)$) is globally asymptotically stable. Therefore by the discussions on the existence/nonexistence of semi-trivial equilibria in Lemmas \ref{dzq}-\ref{dzq2}, the nonexistence of positive equilibrium in Lemmas \ref{23no}-\ref{iv21} and the local stability/instability of semi-trivial equilibria in Lemma \ref{wdx},  the desired results hold. 
\end{proof}

Theorems \ref{dl} and \ref{dl0} show that the drift rate affects  the global dynamics of the two competing species competiton model. If  $(\textbf{H1})$ holds, the species with faster random dispersal rate has competitive advantages. In contrast if  $(\textbf{H2})$ holds, the slower one is superior.


\section{Eigenvalue problems}

 For a real vector $\bm u=(u_1,\cdots,u_n)^T$, we write $\bm u\ge\bm0$ $(\bm u\gg \bm 0)$ if $u_i\ge 0$ ($u_i>0$) for all $1\le i\le n$, and $\bm u>\bm0$ if $\bm u\ge \bm0$ and $\bm u\neq \bm0$. Suppose that $A$ is an $n\times n$ real matrix. The spectral bound $s(A)$ of $A$ is defined to be
\begin{equation*}
s(A)=\max\left\{{\rm Re} \lambda:\;\lambda\;\text{is}\;\text{an}\;\text{eigenvalue}\;\text{of}\;A\right\}.
\end{equation*}
The matrix $A$ is \emph{essentially nonnegative} if all the off-diagonal elements are nonnegative.
By the Perron-Frobenius theorem, if $A$ is an irreducible essentially nonnegative matrix, then $s(A)$ is the unique eigenvalue  (called the principal eigenvalue) of $A$ corresponding to a nonnegative eigenvector.

Let  ${\bm m}=(m_1,\cdots,m_n)^T$ be a real vector, $d>0$ and $\bm q\ge \bm 0$. Consider the following eigenvalue problem:
\begin{equation}\label{eigen}
\sum_{j=1}^{n}(dD_{ij}+q_jQ_{ij})\phi_j+m_i\phi_i=\la\phi_i, \;\;i=1,\cdots,n.
\end{equation}
Since $dD+Q\text{diag}(q_i)+\text{diag}(m_i)$ is irreducible and essentially nonnegative,
\begin{equation}\label{la1}
\la_1(d,\bm q,\bm m):=s(dD+Q\text{diag}(q_i)+\text{diag}(m_i))
\end{equation}
is the principal eigenvalue of $dD+Q\text{diag}(q_i)+\text{diag}(m_i)$. The eigenvalue $\la_1(d,\bm q,\bm m)$ plays an essential role in our analysis. In the rest of this section, we will present several results on the eigenvalue $\la_1(d,\bm q,\bm m)$, which will be used later.



We start by considering the limits of $\la_1(d,\bm q,\bm m)$ as $d\to0$ or $\infty$.

\begin{lemma}\label{dlimit}
Suppose that $d>0$ and $\bm m$ is a real vector. Let $\la_1(d,\bm q,\bm m)$ be defined as above.   Then the following equations hold:
\begin{equation*}
\lim_{d\to0}\la_1(d,\bm q,\bm m)=\max_{1\le i\le n}\{m_i-q_i\}\;\;\text{and}\;\;\lim_{d\to\infty}\la_1(d,\bm q,\bm m)=\f{\sum^n_{i=1}m_i-q_n}{n}.
\end{equation*}
\end{lemma}
\begin{proof}
Clearly,
\begin{equation*}
\lim_{d\to0}\la_1(d,\bm q,\bm m)=\la_1(0,\bm q,\bm m)=\max_{1\le i\le n}\{m_i-q_i\}.
\end{equation*}
It remains to consider the limit of $\la_1(d,\bm q,\bm m)$ as $d\to\infty$. Let $\bm\phi=(\phi_1,\cdots,\phi_n)^T\gg\bm 0$ be the eigenvector corresponding to the eigenvalue $\la_1(d,\bm q,\bm m)$ with $\sum^n_{i=1}\phi_i=1$. By \eqref{eigen}, we have
\begin{equation}\label{egi0}
\sum_{i=1}^{n}\sum_{j=1}^{n}(dD_{ij}+q_jQ_{ij})\phi_j+\sum_{i=1}^{n}m_i\phi_i=\la_1(d,\bm q,\bm m).
\end{equation}
Since
\begin{equation*}
\sum_{i=1}^{n}\sum_{j=1}^{n}D_{ij}\phi_j=0\;\;\text{and}\;\;\sum_{i=1}^{n}\sum_{j=1}^{n}q_jQ_{ij}\phi_j=-q_n\phi_n,
\end{equation*}
we see from \eqref{egi0} that
\begin{equation}\label{egi}
-q_n\phi_n+\sum_{i=1}^{n}m_i\phi_i=\la_1(d,\bm q,\bm m).
\end{equation}
This yields
\begin{equation*}
\min_{1\le i\le n}m_i-q_n\le\la_1(d,\bm q,\bm m)\le\max_{1\le i\le n}m_i.
\end{equation*}
So, up to a subsequence if necessary, we may assume
$$
\lim_{d\to\infty}\la_1(d,\bm q,\bm m)=a\;\;\text{and}\;\;\lim_{d\to\infty}\bm\phi=\hat{\bm \phi},
$$
for some $a\in\mathbb{R}$ and  $\hat{\bm \phi}=\left(\hat{\phi}_1,\cdots,\hat{\phi}_n\right)^T\ge\bm 0$ with $\sum^n_{i=1}\hat{\phi}_i=1$. Dividing both sides of \eqref{eigen} by $d$ and taking $d\to\infty$, we obtain $D\bm \hat{\phi}=0$, which implies that
\begin{equation*}
\hat{\bm \phi}=\left(\hat{\phi}_1,\cdots,\hat{\phi}_n\right)^T=\left(\f{1}{n},\cdots,\f{1}{n}\right)^T.
\end{equation*}
Taking $d\to\infty$ in \eqref{egi}, we have $a=(\sum_{i=1}^{n}m_i-q_n)/{n}$. This completes the proof.
\end{proof}

Then we explore some further properties of $\la_1(d^*,\bm q,\bm r)$, which will be important in the proof of Lemma \ref{dzq}.

\begin{lemma}\label{dds}
Suppose $\bm r=(r,\cdots,r)^T$ with $r>0$. If $\la_1(d^*,\bm q,\bm r)=0$ for some $d^*>0$, then the following statements hold:
\begin{itemize}
  \item [${\rm (i)}$] If $(\rm{\bf{H1}})$ holds, then
\begin{equation*}
\left.\f{\partial}{\partial d}\la_1(d,\bm q,\bm r)\right|_{d=d^*}>0;
\end{equation*}
  \item [${\rm (ii)}$] If $(\rm{\bf{H2}})$ holds, then
\begin{equation*}
\left.\f{\partial}{\partial d}\la_1(d,\bm q,\bm r)\right|_{d=d^*}<0.
\end{equation*}
\end{itemize}
\end{lemma}
\begin{proof}
Let $\bm\phi=(\phi_1,\cdots,\phi_n)^T\gg\bm 0$ be the eigenvector corresponding to the  eigenvalue $\la_1(d,\bm q,\bm r)$ with $\sum^n_{i=1}\phi_i=1$. Differentiating \eqref{eigen} with respect to $d$, we obtain
\begin{equation}\label{dqd}
\sum_{j=1}^{n}(dD_{ij}+q_jQ_{ij})\f{\partial\phi_j}{\partial d}+\sum_{j=1}^{n}D_{ij}\phi_j+r\f{\partial\phi_i}{\partial d}=\la\f{\partial\phi_i}{\partial d}+\f{\partial\la}{\partial d}\phi_i,\;\;i=1,\cdots,n.
\end{equation}
Multiplying \eqref{dqd} by $\phi_i$ and \eqref{eigen} by $\f{\partial\phi_i}{\partial d}$ and taking the difference of them, we have
\begin{equation}\label{dxj0}
\sum_{j\neq i}(dD_{ij}+q_jQ_{ij})
\left(\f{\partial\phi_j}{\partial d}\phi_i-\f{\partial\phi_i}{\partial d}\phi_j\right)
+\sum_{j=1}^{n}D_{ij}\phi_i\phi_j=\f{\partial\la}{\partial d}\phi_i^2,\;\;i=1,\cdots,n.
\end{equation}
Let
\begin{equation*}
\rho_i=\prod_{k=0}^{i-1}\f{d}{d+q_k},\;\;\;\;i=1,\cdots,n.
\end{equation*}
Multiplying \eqref{dxj0} by $\rho_i$ and summing up over $i$, we get
\begin{equation}\label{dxj}
\f{\partial\la}{\partial d}\sum_{i=1}^{n}\rho_i\phi_i^2
=\sum_{i=1}^{n}\sum_{j\neq i}\rho_i(dD_{ij}+q_jQ_{ij})
\left(\f{\partial\phi_j}{\partial d}\phi_i-\f{\partial\phi_i}{\partial d}\phi_j\right)+\sum_{i,j=1}^{n}\rho_iD_{ij}\phi_i\phi_j.
\end{equation}
It is not difficult to check that  $\left(\rho_i(dD_{ij}+q_jQ_{ij})\right)$ is symmetric. Therefore, we have
$$ 
\sum_{i=1}^{n}\sum_{j\neq i}\rho_i(dD_{ij}+q_jQ_{ij})
\left(\f{\partial\phi_j}{\partial d}\phi_i-\f{\partial\phi_i}{\partial d}\phi_j\right)=0.
$$
It follows that 
\begin{equation}\label{dxj1}
\f{\partial\la}{\partial d}\sum_{i=1}^{n}\rho_i\phi_i^2
=\sum_{i,j=1}^{n}\rho_iD_{ij}\phi_i\phi_j.
\end{equation}

Let $\bm\phi^*=(\phi_1^*,\dots, \phi_n^*)^T$ be the positive eigenvector corresponding to $\la_1(d^*,\bm q,\bm r)=0$  with $\sum^n_{i=1}\phi^*_i=1$ and 
$$
\rho^*_i=\prod_{k=0}^{i-1}\f{d^*}{d^*+q_k}, \ \ \ i=1, \dots, n.
$$
Evaluating \eqref{dxj} at $d=d^*$ leads to
\begin{equation}\label{fzjs}
\begin{split}
\left.\f{\partial\la}{\partial d}\right|_{d=d^*}\sum_{i=1}^{n}\rho^*_i(\phi_i^*)^2=&\sum_{i,j=1}^{n}\rho^*_iD_{ij}\phi^*_i\phi^*_j\\
=&\sum_{i=1}^{n-1}\rho^*_i\phi^*_i(\phi^*_{i+1}-\phi^*_i)+
\sum_{i=2}^{n}\rho^*_i\phi^*_i(\phi^*_{i-1}-\phi^*_i)\\
=&\f{1}{d^*}\sum_{i=1}^{n-1}\rho^*_{i+1}(\phi^*_{i+1}-\phi^*_i)\left[(d^*+q_i)\phi^*_i-d^*\phi^*_{i+1}\right].
\end{split}
\end{equation}

By \eqref{eigen},  if $n\ge3$ then $\bm \phi^*$ satisfies
\begin{subequations}\label{phi0}
\begin{align}
&d^*(\phi^*_2-\phi^*_1)=-\phi^*_1(r-q_1),\label{phi0-a}\\
&d^*(\phi^*_{i+1}-\phi^*_i)-(d^*+q_{i-1})(\phi^*_i-\phi^*_{i-1})=-\phi^*_i(r+q_{i-1}-q_i),\;i=2,\cdots,n-1,\label{phi0-b}\\
&-(d^*+q_{n-1})(\phi^*_n-\phi^*_{n-1})=-\phi^*_n(r+q_{n-1}-q_n).\label{phi0-c}
\end{align}
\end{subequations}
If $n=2$, then $\bm \phi^*$ satisfies only \eqref{phi0-a} and \eqref{phi0-c}.

Now we consider case (i) and suppose that assumption $(\textbf{H1})$ holds.  We claim  $\phi^*_2>\phi^*_1$. To see it, suppose to the contrary that $\phi^*_2\le \phi^*_1$. By \eqref{phi0-a}, we have  $r\ge q_1$. Then by $(\textbf{H1})$, we have
\begin{equation}\label{h}
r\ge q_1=q_1-q_0\ge\cdots \ge q_n-q_{n-1},
\end{equation}
with at lease one strict inequality. 
If $n\ge 3$, by \eqref{phi0-b} and induction, we obtain that
\begin{equation}\label{phijs}
\phi^*_1\ge \cdots\ge \phi^*_n.
\end{equation}
If $n=2$, \eqref{phijs} holds trivially. By \eqref{phijs}, \eqref{phi0-c} and \eqref{h}, we have
\begin{equation*}
0\le-(d^*+q_{n-1})(\phi^*_n-\phi^*_{n-1})=-\phi^*_n(r+q_{n-1}-q_n)<0,
\end{equation*}
which is a contradiction. Therefore, $\phi^*_2>\phi^*_1$. Similarly, we can show that
\begin{equation}\label{phiddx}
\phi^*_{i+1}>\phi^*_i,\;\;\;\;i=1,\cdots,n-1.
\end{equation}

We  rewrite \eqref{eigen} as follows:
\begin{subequations}\label{phi1}
\begin{align}
&-[(d^*+q_1)\phi^*_1-d^*\phi^*_2]=-r\phi^*_1,\label{phi1-a}\\
&[(d^*+q_{i-1})\phi^*_{i-1}-d^*\phi^*_i]-[(d^*+q_i)\phi^*_i-d^*\phi^*_{i+1}]=-r\phi^*_i,\;i=2,\cdots,n-1,\label{phi1-b}\\
&[(d^*+q_{n-1})\phi^*_{n-1}-d^*\phi^*_n]=-(r-q_n)\phi^*_n.\label{phi1-c}
\end{align}
\end{subequations}
By \eqref{phi1-a}, we have $(d^*+q_1)\phi^*_1-d^*\phi^*_2>0$. Then by \eqref{phi1}, we can show that
\begin{equation}\label{phif}
(d^*+q_i)\phi^*_i-d^*\phi^*_{i+1}>0,\;\; \;\;i=1,\cdots,n-1.
\end{equation}

Combining \eqref{fzjs}, \eqref{phiddx} and \eqref{phif}, we have
\begin{equation*}
\left.\f{\partial}{\partial d}\la_1(d,\bm q,\bm r)\right|_{d=d^*}>0.
\end{equation*}
This proves (i). Case (ii) can be prove similarly. Indeed, we can use a similar argument as above to prove $\phi_i^*>\phi_{i+1}^*$ for $i=1, \dots, n-1$, which trivially implies $(d^*+q_i)\phi^*_i-d^*\phi^*_{i+1}>0$ for $i=1, \dots, n-1$. By these and  \eqref{fzjs}, (ii) holds. 
\end{proof}

If $\{m_i\}_{i=1}^n$ is monotone decreasing, we have the following result about the properties of the eigenvector corresponding to $\la_1(d,\bm q,\bm m)$.
\begin{lemma}\label{prop-p11}
Suppose that $d>0$, $\bm q\ge \bm 0$, and the components of $\bm m=(m_1, \dots, m_n)$ satisfy
 $m_1\ge \cdots \ge m_n$ with at least one strict inequality. Let $\bm {\phi}=(\phi_1,\cdots,\phi_n)^T\gg \bm 0$ be an eigenvector of \eqref{eigen} corresponding to the principal eigenvalue $\la_1(d,\bm q,\bm m)$. Then $\bm \phi$ satisfies
\begin{equation}\label{ephi}
(d+q_i)\phi_{i}-d\phi_{i+1}>0, \;\;\;\;i=1,\cdots,n-1.
\end{equation}
\end{lemma}
\begin{proof}
If $n\ge3$, then $\bm \phi$ satisfies
\begin{subequations}\label{eigg}
\begin{align}
& -(d+q_1)\phi_1+d\phi_2=(\lambda_1-m_1)\phi_1,\label{eigg-a}\\
&\left[(d+q_{i-1})\phi_{i-1}-d\phi_i\right]-\left[(d+q_i)\phi_{i}-d\phi_{i+1}\right]=(\lambda_1-m_i)\phi_i,\;\;i=2,\cdots,n-1,\label{eigg-b}\\
&(d+q_{n-1})\phi_{n-1}-d\phi_n=(\lambda_1+q_n-m_n)\phi_n.\label{eigg-c}
\end{align}
\end{subequations}
If $n=2$, $\bm \phi$ satisfies only \eqref{eigg-a} and \eqref{eigg-c}. We first claim that $(d+q_1)\phi_1-d\phi_2>0$. Suppose to the contrary that $(d+q_1)\phi_1-d\phi_2
\le 0$. Then by \eqref{eigg-a}, we have $\la_1-m_1\ge0$. Since $m_1\ge \cdots \ge m_n$ with at least one strict inequality, we obtain that $\la_1-m_i\ge0$ for $i=1,\cdots,n-1$ and $\la_1+q_n-m_n>0$. Then by \eqref{eigg-b} and induction, we  can deduce that
\begin{equation*}
(d+q_i)\phi_i-d\phi_{i+1}\le0,\;\;\;\;i=1,\cdots,n-1.
\end{equation*}
Therefore, by \eqref{eigg-c}, we obtain
\begin{equation}\label{contrd}
0\ge(d+q_{n-1})\phi_{n-1}-d\phi_n=(\lambda_1+q_n-m_n)\phi_n>0,
\end{equation}
which is a contradiction. Thus, $(d+q_1)\phi_1-d\phi_2>0$. Applying similar arguments to \eqref{eigg-a}-\eqref{eigg-c}, we can prove \eqref{ephi}.
\end{proof}

The following result is used in the proof of  Lemma \ref{wdx} later.
\begin{lemma}\label{prpla}
Suppose that $d_1, d_2>0$, $\bm q\ge \bm 0$, and $\bm m$ is a real vector. Let $\bm\phi_j=(\phi_{j,1},\cdots,\phi_{j,n})^T\gg\bm 0$ be an eigenvector of \eqref{eigen} corresponding to the principal eigenvalue $\la_1(d_j,\bm q,\bm m)$ for $j=1,2$. Then, the following equation holds:
\begin{equation}\label{iden}
\begin{split}
&\left[\la_1(d_1,\bm q,\bm m)-\la_1(d_2,\bm q,\bm m)\right]
\sum_{k=1}^n\rho^{(2)}_k\phi_{1,k}\phi_{2,k}\\
=&\f{(d_2-d_1)}{d_2}\sum_{k=1}^{n-1}\rho^{(2)}_{k+1}
\left(\phi_{1,k+1}-\phi_{1,k}\right)\left[d_2\phi_{2,k+1}-(d_2+q_k)\phi_{2,k}\right],
\end{split}
\end{equation}
where $\rho^{(2)}_0=1$ and
\begin{equation}\label{rho2}
\rho^{(2)}_k=\prod_{i=1}^{k-1}\f{d_2}{d_2+q_i},\;\;\;\;i=2,\cdots,n.
\end{equation}
\end{lemma}
\begin{proof}
Denote $\la_1^j=\la_1(d_j,\bm q,\bm m)$, $j=1, 2$. Let
\begin{equation}\label{phi}
\Phi^j_0=0,\;\;\Phi^j_n=-q_n\phi_{j,n},\;\;\Phi^j_k=d_2\phi_{j,k+1}-(d_2+q_k)\phi_{j,k},\;\;k=1,\cdots,n-1,
\end{equation}
and
\begin{equation}\label{psi0}
\Psi_0=\Psi_n=0,\;\;\Psi_k=(d_2-d_1)\left(\phi_{1,k+1}-\phi_{1,k}\right),\;\;k=1,\cdots,n-1.
\end{equation}
Then we see from \eqref{eigen} that
\begin{equation}\label{iden1}
\Phi^1_k-\Phi^1_{k-1}+\left(m_k-\la^1_1\right)\phi_{1,k}=\Psi_k-\Psi_{k-1},\;\;k=1,\cdots,n.
\end{equation}
Multiplying \eqref{iden1} by $\rho^{(2)}_k\phi_{2,k}$ and summing  over  $k$, we obtain
\begin{equation}\label{iden2}
\begin{split}
&\sum^n_{k=1}\rho^{(2)}_k\left(\Phi^1_k-\Phi^1_{k-1}\right)\phi_{2,k}\\
=&\sum^n_{k=1}\rho^{(2)}_k\left(\Psi_k-\Psi_{k-1}\right)\phi_{2,k}
-\sum^n_{k=1}\rho^{(2)}_k\left(m_k-\la^1_1\right)\phi_{1,k}\phi_{2,k}.\\
\end{split}
\end{equation}
A direct computation yields
\begin{equation}\label{iden2i1}
\begin{split}
&\sum^n_{k=1}\rho^{(2)}_k\left(\Phi^1_k-\Phi^1_{k-1}\right)\phi_{2,k}\\
=&-\rho^{(2)}_1\Phi^1_0\phi_{2,1}+\sum^{n-1}_{k=1}\rho^{(2)}_k\Phi^1_k\left(\phi_{2,k}-\f{d_2}{d_2+q_k}\phi_{2,k+1}\right)
+\rho^{(2)}_n\Phi^1_n\phi_{2,n}\\
=&-\frac{1}{d_2}\sum^{n-1}_{k=1}\rho^{(2)}_{k+1}\Phi^1_k\Phi_k^2
-q_n\rho^{(2)}_n\phi_{1,n}\phi_{2,n},\\
\end{split}
\end{equation}
where we have used \eqref{phi} in the last step. Then we compute
\begin{equation}\label{iden2i2}
\begin{split}
&\sum^n_{k=1}\rho^{(2)}_k\left(\Psi_k-\Psi_{k-1}\right)\phi_{2,k}\\
=&-\rho^{(2)}_1\Psi_0\phi_{2,1}+\sum^{n-1}_{k=1}\rho^{(2)}_k\Psi_k\left(\phi_{2,k}-\f{d_2}{d_2+q_k}\phi_{2,k+1}\right)
+\rho^{(2)}_n\Psi_n\phi_{2,n}\\
=&-\frac{1}{d_2}\sum^{n-1}_{k=1}\rho^{(2)}_{k+1}\Psi_k\Phi_k^2,
\end{split}
\end{equation}
where we have used \eqref{psi0} in the last step. By \eqref{iden2}-\eqref{iden2i2}, we have
\begin{equation}\label{iden3}
\begin{split}
&\sum^n_{k=1}\rho^{(2)}_k\left(m_k-\la^1_1\right)\phi_{1,k}\phi_{2,k}\\
=&-\f{1}{d_2}\sum^{n-1}_{k=1}\rho^{(2)}_{k+1}\Psi_k\Phi^2_k
+\f{1}{d_2}\sum^{n-1}_{k=1}\rho^{(2)}_{k+1}\Phi^1_k\Phi^2_k
+q_n\rho^{(2)}_n\phi_{1,n}\phi_{2,n},
\end{split}
\end{equation}
Similarly, by \eqref{eigen}, we have
\begin{equation}\label{1iden}
\Phi^2_k-\Phi^2_{k-1}+\left(m_k-\la^2_1\right)\phi_{2,k}=0,\;\;k=1,\cdots,n.
\end{equation}
Multiplying \eqref{1iden} by $\rho^{(2)}_k\phi_{1,k}$ and summing  over  $k$, we obtain
\begin{equation}\label{2iden}
\begin{split}
&\sum^n_{k=1}\rho^{(2)}_k\left(m_k-\la^2_1\right)\phi_{1,k}\phi_{2,k}\\
=&-\sum^{n-1}_{k=1}\rho^{(2)}_k\Phi^2_k\left(\phi_{1,k}-\f{d_2}{d_2+q_k}\phi_{1,k+1}\right)
+q_n\rho^{(2)}_n\phi_{1,n}\phi_{2,n}\\
=&\f{1}{d_2}\sum^{n-1}_{k=1}\rho^{(2)}_{k+1}\Phi^1_k\Phi^2_k
+q_n\rho^{(2)}_n\phi_{1,n}\phi_{2,n}.\\
\end{split}
\end{equation}
Subtracting \eqref{iden3} from \eqref{2iden}, we obtain  \eqref{iden}.
\end{proof}

\begin{remark}
We remark that Lemmas \ref{prop-p11} and \ref{prpla} have been proved in \cite{chen2022} for a special case when $q_1=q_2=\cdots=q_n$.
\end{remark}

\section{Existence and properties of semi-trivial equilibria}
To study the existence and properties of the semi-trivial equilibria of \eqref{pat-cp2}, we  need to consider the positive equilibrium of the following:
 \begin{equation}\label{pat-cp3}
 \left\{
 \begin{array}{lll}
\ds\frac{dw_i}{dt}=\sum_{j=1}^{n} (d D_{ij}+q_jQ_{ij})w_j+w_i(r-w_i), \;\;i=1,\cdots,n, \ t>0,\\
\bm w(0)>\bm 0.
\end{array}
\right.
\end{equation}
The global dynamics of \eqref{pat-cp3} as stated in the following result is well-known:
\begin{lemma}[\cite{cosner1996variability,Li2010,Lu1993}]\label{DS-single}
Suppose that $d>0$, $\bm q\ge \bm 0$, and $\bm r =(r,\cdots,r)^T$ with $r>0$. Then if $\la_1(d, \bm q, \bm r)>0$, \eqref{pat-cp3} has a unique positive equilibrium which is globally asymptotically stable; if $\la_1(d, \bm q, \bm r)\le0$, the trivial equilibrium $\bm 0$ of \eqref{pat-cp3} is globally asymptotically stable. 
\end{lemma}

By Lemmas \ref{dlimit}-\ref{dds} and \ref{pat-cp3}, we have the following two results about the existence/nonexistence of positive equilibrium of \eqref{pat-cp3}.
\begin{lemma}\label{dzq}
Suppose that $\bm q\ge \bm 0$, $(\rm{\bf{H1}})$ holds, and $\bm r =(r,\cdots,r)^T$ with $r>0$. Then the following statements holds:
\begin{enumerate}
\item [${\rm (i)}$]  If $\ds\min_{1\le i\le n}q_i<r$, then \eqref{pat-cp3} admits a unique positive equilibrium for any $d>0$, which is globally asymptotically stable;
\item [${\rm (ii)}$] If $\ds\min_{1\le i\le n}q_i>r$ and $q_n<rn$, then there exists $d^*(\bm q, \bm r)>0$ such that $\la_1(d^*,\bm q,\bm r)=0$, $\la_1(d,\bm q,\bm r)<0$ for $0<d<d^*$, and $\la_1(d,\bm q,\bm r)>0$ for $d>d^*$. Moreover, the following results hold:
    \begin{enumerate}
    \item [${\rm (ii_1)}$] If $d\in(d^*,\infty)$, then \eqref{pat-cp3} admits a unique positive equilibrium, which is globally asymptotically stable;
    \item [${\rm (ii_2)}$] If $d\in(0,d^*]$, then the trivial equilibrium $\bm 0$ of \eqref{pat-cp3} is globally asymptotically stable;
    \end{enumerate}
\item [${\rm (iii)}$] If $\ds\min_{1\le i\le n}q_i>r$ and $q_n>rn$, then the trivial equilibrium $\bm 0$ of \eqref{pat-cp3} is globally asymptotically stable for any $d>0$.
\end{enumerate}
\end{lemma}
\begin{proof}
(i) By $\ds\min_{1\le i\le n}q_i<r$ and Lemma \ref{dlimit}, we have $\lim_{d\to0}\la_1(d,\bm q,\bm m)>0$ and $\lim_{d\to\infty}\la_1(d,\bm q,\bm m)>0$. By Lemma \ref{dds} (i), we have $\la_1(d,\bm q,\bm r)>0$  for all $d>0$. Therefore by Lemma \ref{DS-single}, \eqref{pat-cp3} has a unique positive equilibrium, which is globally asymptotically stable. The proof of (ii)-(iii) is similar, so it is omitted here. 
\end{proof}

If we replace (\rm{\bf{H1}}) by (\rm{\bf{H2}}) in Lemma \ref{dzq}, we have the following result:
\begin{lemma}\label{dzq2}
Suppose that $\bm q$ satisfies $(\rm{\bf{H2}})$ and $\bm r =(r,\cdots,r)^T$ with $r>0$. Then the following statements holds:
\begin{enumerate}
\item [${\rm (i)}$]  If $q_n<rn$, then \eqref{pat-cp3} admits a unique positive equilibrium for any $d>0$, which is globally asymptotically stable;
\item [${\rm (ii)}$] If $q_n>rn$ and $q_1<r$, then there exists $\widetilde d^*(\bm q, \bm r)>0$ such that $\la_1(\widetilde d^*,\bm q,\bm r)=0$, $\la_1(d,\bm q,\bm r)>0$ for $0<d<\widetilde d^*$, and $\la_1(d,\bm q,\bm r)<0$ for $d>\widetilde d^*$. Moreover,  the following results hold:
    \begin{enumerate}
    \item [${\rm (ii_1)}$] If $d\in(0,\widetilde d^*)$, then \eqref{pat-cp3} admits a unique positive equilibrium, which is globally asymptotically stable solution;
    \item [${\rm (ii_2)}$] If $d\in[\widetilde d^*,\infty)$, then the trivial equilibrium $\bm 0$ of \eqref{pat-cp3} is globally asymptotically stable;
    \end{enumerate}
\item [${\rm (iii)}$] If $q_n>rn$ and $q_1>r$, then the trivial equilibrium $\bm 0$ of \eqref{pat-cp3} is globally asymptotically stable for any $d>0$.
\end{enumerate}
\end{lemma}


Then we prove some properties of the positive equilibrium of  \eqref{pat-cp3}, which will be useful later. 
\begin{lemma}\label{wdx0}
Suppose that  $\bm q\ge \bm 0$ and $d, r>0$. Let  $\bm w^*=(w_1,\cdots,w_n)^T\gg0$ be the positive equilirbium of \eqref{pat-cp3} if exists. Then the following statements hold:
\begin{enumerate}
    \item[$\rm{(i)}$] If $(\rm{\bf{H1}})$ holds, then $w_1^*<\cdots<w_n^*$ and $(d+q_i)w^*_{i}>dw^*_{i+1}$ for $i=1,\cdots,n-1$;
    \item[$\rm{(ii)}$] If $(\rm{\bf{H2}})$ holds, then $w_1^*>\cdots>w_n^*$ and $(d+q_i)w^*_{i}>dw^*_{i+1}$ for $i=1,\cdots,n-1$.
\end{enumerate}
\end{lemma}
\begin{proof}
We first prove (i).
By \eqref{pat-cp3}, if $n\ge3$, then $\bm w^*$ satisfies
\begin{subequations}\label{bpf}
\begin{align}
&d(w^*_2-w^*_1)=-w^*_1(r-q_1-w^*_1),\label{bpf-a}\\
&d(w^*_{i+1}-w^*_i)-(d+q_{i-1})(w^*_i-w^*_{i-1})=-w^*_i(r+q_{i-1}-q_i-w^*_i),\;i=2,\cdots,n-1,\label{bpf-b}\\
&-(d+q_{n-1})(w^*_n-w^*_{n-1})=-w^*_n(r+q_{n-1}-q_n-w^*_n).\label{bpf-c}
\end{align}
\end{subequations}
If $n=2$, then $\bm w^*$ only satisfies \eqref{bpf-a} and \eqref{bpf-c}.

Suppose to the contrary that $w^*_2\le w^*_1$. Then we see from \eqref{bpf-a} that $r-q_1-w^*_1\ge0$. If $n=2$, by $(\textbf{H1})$, we have $r+q_1-q_2-w^*_2>0$. This together with \eqref{bpf-c} implies that
\begin{equation*}
0\le-(d+q_1)(w^*_2-w^*_1)=-w^*_2(r+q_1-q_2-w^*_2)<0,
\end{equation*}
which is a contradiction. If $n\ge3$, then by \eqref{bpf-b},  $(\textbf{H1})$, and induction, we can show
\begin{equation}\label{inq1}
r+q_{n-1}-q_n-w^*_n>0\;\;\text{and}\;\;w^*_2\ge \cdots\ge w^*_n.
\end{equation}
So, by \eqref{bpf-c} and \eqref{inq1}, we have
\begin{equation}\label{md}
0\le-(d+q_{n-1})(w^*_n-w^*_{n-1})=-w^*_n(r+q_{n-1}-q_n-w^*_n)<0,
\end{equation}
which is a contradiction. Thus, $w^*_2>w^*_1$.

Suppose to the contrary that $w^*_3\le w^*_2$. Since $w^*_2>w^*_1$, we see from \eqref{bpf-b} that
$r+q_1-q_2-w^*_2>0$.
Then, by  $(\textbf{H1})$, \eqref{bpf-b}, and induction, we can show
\begin{equation*}
w^*_2\ge w^*_3\ge \cdots\ge w^*_n\;\;\text{and}\;\;r+q_{n-1}-q_n-w^*_n>0,
\end{equation*}
which leads to a contradiction as \eqref{md}. Continuing this process, we obtain 
\begin{equation}\label{uddx}
w^*_{i+1}>w^*_i\;\;\text{for}\;\;i=1,\cdots,n-1.
\end{equation}

By \eqref{uddx}, we have $r-w^*_1>\cdots>r-w^*_n$.
Noticing that $\bm w^*$ is an eigenvector corresponding to eigenvalue $\la_1(d,\bm q,\bm r-\bm w^*)=0$, 
by Lemma \ref{prop-p11}, we have 
$$
(d+q_i)w^*_{i}>dw^*_{i+1}, \ \ \ i=1, \dots, n-1.
$$

Now we consider (ii). Using similar arguments as (i), we can show that $w_1^*>\cdots>w_n^*$. This trivially yields $(d+q_i)w^*_{i}>dw^*_{i+1}$ for $i=1,\cdots,n-1$.
\end{proof}

\section{Nonexistence of positive equilibrium}

In this section, we prove the nonexistence of positive equilibrium of model \eqref{pat-cp2}, which is an essential step towards understanding the global dynamics of \eqref{pat-cp2}. 

Suppose that $\bm q\ge \bm 0$ and $r>0$. Let $(\bm u,\bm v)$ be a positive equilibrium of \eqref{pat-cp2} if exists. Define
\begin{subequations}\label{figi}
\begin{align}
&f_0=f_n=0,\;\;f_k=d_1u_{k+1}-(d_1+q_k)u_k,\;\;\;\;k=1,\cdots,n-1,\label{figi-f}\\
&g_0=g_n=0,\;\;g_k=d_2v_{k+1}-(d_2+q_k)v_k,\;\;\;\;k=1,\cdots,n-1.\label{figi-g}
\end{align}
\end{subequations}
By \eqref{pat-cp2}, we have
\begin{subequations}\label{f-k}
\begin{align}
&f_{k}-f_{k-1}=-u_k(r-u_k-v_k) ,\;\;\;k=1,\cdots,n-1,\label{f-k-a}\\
&f_{n}-f_{n-1}=-u_n(r-q_n-u_n-v_n),\label{f-k-b}
\end{align}
\end{subequations}
and
\begin{subequations}\label{g-k}
\begin{align}
&g_{k}-g_{k-1}=-v_k(r-u_k-v_k) ,\;\;\;k=1,\cdots,n-1,\label{g-k-a}\\
&g_{n}-g_{n-1}=-v_n(r-q_n-u_n-v_n).\label{g-k-b}
\end{align}
\end{subequations}
\begin{lemma}\label{proffg}
 Let $\{f_k\}_{k=0}^{n}$ and $\{g_k\}_{k=0}^{n}$ be defined in \eqref{figi}. Then $f_1,g_1,f_{n-1},g_{n-1}<0$.
\end{lemma}
\begin{proof}
Suppose to the contrary that $f_1\ge0$. By \eqref{f-k-a}, we have $r-u_1-v_1\le0$. This, combined with \eqref{g-k-a}, yields $g_1\ge0$. Noticing  $f_1,g_1\ge0$, we see
from \eqref{figi} that $u_2>u_1$,  $v_2>v_1$, and $r-u_2-v_2<0$. If $n=2$, then
\begin{equation*}
0\ge-f_1=-u_2(r-q_2-u_2-v_2)>0,
\end{equation*}
which is a contradiction.
If $n\ge 3$, then by \eqref{f-k}-\eqref{g-k} and induction, we have $f_k,g_k>0$ for $k=2,\cdots,n-1$.  Consequently, we obtain that
\begin{equation*}
\begin{split}
&u_{k}>u_{k-1},\;\;v_{k}>v_{k-1},\;\;\text{for}\;\;k=3,\cdots,n,\\
&r-u_k-v_k<0\;\;\text{for}\;\;k=3,\cdots,n-1,\;\;\text{and}\;\;r-q_n-u_n-v_n<0,
\end{split}
\end{equation*}
which contradicts \eqref{f-k}-\eqref{g-k} with $k=n$. Thus $f_1<0$, which yields $g_1<0$. Similarly, we can prove that $f_{n-1},g_{n-1}<0$, and here we omit the details. This completes the proof.
\end{proof}

Then we define another two auxiliary sequences $\{T_k\}_{k=0}^n$ and $\{S_k\}_{k=0}^n$:
\begin{subequations}\label{tsk}
\begin{align}
&T_0=T_n=0,\;\;T_k=u_{k+1}-u_k\;\;\text{for}\;\;k=1,\cdots,n-1,\label{tsk-t}\\
&S_0=S_n=0,\;\;S_k=v_{k+1}-v_k\;\;\text{for}\;\;k=1,\cdots,n-1.\label{tsk-s}
\end{align}
\end{subequations}

\begin{lemma}\label{indfg}
Let $\{T_k\}_{k=0}^n$ and $\{S_k\}_{k=0}^n$ be defined in \eqref{tsk}, and let $\{f_k\}_{k=0}^{n}$ and $\{g_k\}_{k=0}^{n}$ be defined in \eqref{figi}. Then
for any $1\le i<j\le n$, the following two identities hold:
\begin{equation}\label{gf2}
\begin{split}
&\ds\f{d_2-d_1}{d_2}\sum_{k=i}^{j-1}\rho^{(2)}_{k+1}T_kg_k=\rho^{(2)}_j\left(d_2u_jS_{j}-d_1v_jT_j\right)\\
&-\rho^{(2)}_i\left[(d_2+q_{i-1})u_iS_{i-1}-(d_1+q_{i-1})v_iT_{i-1}\right],\\
\end{split}
\end{equation}
and
\begin{equation}\label{fg2}
\begin{split}
&\ds\f{d_1-d_2}{d_1}\sum_{k=i}^{j-1}\rho^{(1)}_{k+1}S_kf_k=\rho^{(1)}_j\left(d_1v_jT_{j}-d_2u_jS_j\right)\\
&-\rho^{(1)}_i\left[(d_1+q_{i-1})v_iT_{i-1}-(d_2+q_{i-1})u_iS_{i-1}\right],\\
\end{split}
\end{equation}
where $\rho_{1}^{(1)}=1$ and
\begin{equation}\label{rho1}
\rho_{k}^{(1)}=\prod_{l=1}^{k-1}\f{d_1}{d_1+q_l},\;\;\;\;k=2,\cdots,n,
\end{equation}
 and $\ds\{\rho_{k}^{(2)}\ds\}_{k=1}^n$ is defined in Lemma \ref{prpla}.
\end{lemma}
\begin{proof}
We only prove  \eqref{gf2} since \eqref{fg2} can be proved similarly. For any $k=1,\cdots,n$, by \eqref{pat-cp2}, we have
\begin{equation}\label{tk}
d_2T_{k}-(d_2+q_{k-1})T_{k-1}=-u_k(r+q_{k-1}-q_k-u_k-v_k)+(d_2-d_1)(T_k-T_{k-1}),
\end{equation}
and
\begin{equation}\label{sk}
d_2S_{k}-(d_2+q_{k-1})S_{k-1}=-v_k(r+q_{k-1}-q_k-u_k-v_k).
\end{equation}
Multiplying \eqref{tk} by $\rho^{(2)}_kv_k$ and summing up from $k=i$ to $k=j$, we have
\begin{equation}\label{sumf1}
\begin{split}
&\sum_{k=i}^j\rho^{(2)}_k\left[d_2T_{k}-(d_2+q_{k-1})T_{k-1}\right]v_k\\
=&(d_2-d_1)\sum_{k=i}^j\rho^{(2)}_k\left(T_k-T_{k-1}\right)v_k-\sum_{k=i}^j \rho^{(2)}_ku_kv_k(r+q_{k-1}-q_k-u_k-v_k).\\
\end{split}
\end{equation}
A direct computation yields
\begin{equation}\label{sums1}
\begin{split}
&\sum_{k=i}^j\rho^{(2)}_k\left[d_2T_{k}-(d_2+q_{k-1})T_{k-1}\right]v_k\\
=&-\rho^{(2)}_i(d_2+q_{i-1})v_iT_{i-1}-d_2\sum_{k=i}^{j-1}\rho^{(2)}_kT_kS_k+d_2\rho^{(2)}_jv_jT_j,\\
\end{split}
\end{equation}
and
\begin{equation}\label{sums2}
\begin{split}
&\sum_{k=i}^j\rho^{(2)}_k\left(T_k-T_{k-1}\right)v_k\\
=&-\rho^{(2)}_iv_iT_{i-1}-\ds\f{1}{d_2}\sum_{k=i}^{j-1}\rho^{(2)}_{k+1}T_kg_k+\rho^{(2)}_jv_jT_j.
\end{split}
\end{equation}
Substituting \eqref{sums1}-\eqref{sums2} into \eqref{sumf1}, we obtain
\begin{equation}\label{sumf3}
\begin{split}
&-(d_1+q_{i-1})\rho^{(2)}_iv_iT_{i-1}-d_2\sum_{k=i}^{j-1}\rho^{(2)}_kT_kS_k+d_1\rho^{(2)}_jv_jT_j\\
=&-\ds\f{d_2-d_1}{d_2}\sum_{k=i}^{j-1}\rho^{(2)}_{k+1}T_kg_k-\sum_{k=i}^j\rho^{(2)}_k u_kv_k(r+q_{k-1}-q_k-u_k-v_k).
\end{split}
\end{equation}
Similarly, multiplying \eqref{sk} by $\rho^{(2)}_ku_k$ and summing up from $k=i$ to $k=j$, we have
\begin{equation}\label{sumg2}
\begin{split}
&-(d_2+q_{i-1})\rho^{(2)}_iu_iS_{i-1}-d_2\sum_{k=i}^{j-1}\rho^{(2)}_kT_kS_k+d_2\rho^{(2)}_ju_jS_j\\
=&-\sum_{k=i}^j\rho^{(2)}_k u_kv_k(r+q_{k-1}-q_k-u_k-v_k).
\end{split}
\end{equation}
Taking the difference of \eqref{sumf3} and \eqref{sumg2}, we obtain  \eqref{gf2}. 
\end{proof}

In the following, we say that a sequence changes sign if it has both negative and positive terms. 
\begin{lemma}\label{tsbh}
Let $\{T_k\}_{k=0}^n$ and $\{S_k\}_{k=0}^n$ be defined in \eqref{tsk} and suppose that $d_1\ne d_2$. Then the following statements hold:
\begin{enumerate}
\item [${\rm (i)}$] If $(\rm{\bf{H1}})$ holds, then
\begin{enumerate}
\item [${\rm (i_1)}$]  $T_1, S_1, T_{n-1}, S_{n-1}>0$;
\item [${\rm (i_2)}$] If $n\ge4$, then $\{T_k\}_{k=1}^{n-1}$, $\{S_k\}_{k=1}^{n-1}$ must change sign;
\end{enumerate}
\item [${\rm (ii)}$] If $(\rm{\bf{H2}})$ holds, then
\begin{enumerate}
\item [${\rm (ii_1)}$]  $T_1, S_1, T_{n-1}, S_{n-1}<0$;
\item [${\rm (ii_2)}$] If $n\ge4$, then $\{T_k\}_{k=1}^{n-1}$, $\{S_k\}_{k=1}^{n-1}$ must change sign.
\end{enumerate}
\end{enumerate}
\end{lemma}
\begin{proof}
We rewrite \eqref{tk}-\eqref{sk} as follows:
\begin{subequations}\label{TSk}
\begin{align}
&d_1T_{k}-(d_1+q_{k-1})T_{k-1}=-u_k(r+q_{k-1}-q_k-u_k-v_k),\;\;\;\;k=1,\cdots,n,\label{TSk-a}\\
&d_2S_{k}-(d_2+q_{k-1})S_{k-1}=-v_k(r+q_{k-1}-q_k-u_k-v_k),\;\;\;\;k=1,\cdots,n.\label{TSk-b}
\end{align}
\end{subequations}
We first consider (i). Suppose to the contrary that $T_1\le0$. Since $T_0=0$, we see from \eqref{TSk-a} that
\begin{equation*}
r+q_0-q_1-u_1-v_1\ge0,
\end{equation*}
which yields $S_1\le0$ by \eqref{TSk-b}. By $S_1, T_1\le 0$, we have  $u_2\le u_1$ and $v_2\le v_1$. This, combined with assumption $(\textbf{H1})$, yields
\begin{equation*}
r+q_1-q_2-u_2-v_2\ge0.
\end{equation*}
Then by \eqref{TSk} and induction, we can show that
\begin{equation*}\label{uvinq}
T_k\le0,\;\;S_k\le0,\;\;\text{and}\;\;r+q_{k}-q_{k+1}-u_{k+1}-v_{k+1}\ge 0\;\; \text{for}\;\;k=1,\cdots,n-1\;.
\end{equation*}
Moreover, by the assumption that there exists at least one strict inequality in $(\textbf{H1})$, we have $r+q_{n-1}-q_{n}-u_{n}-v_{n}>0$.
Noticing $T_n=0$, we see from  \eqref{TSk-a} that
\begin{equation*}
0\le-(d_1+q_{n-1})T_{n-1}=-u_n(r+q_{n-1}-q_n-u_k-v_k)<0,
\end{equation*}
which is a contradiction. Therefore we have $T_1>0$. It follows from \eqref{TSk} with $k=n$ that $S_1>0$. Using a similar argument, we can show that  $T_{n-1},S_{n-1}>0$.

 Suppose to the contrary that $\{T_k\}_{k=1}^{n-1}$ does not change sign.
Since $T_1,T_{n-1}>0$, we must have $T_k\ge0$ for $k=1,\cdots,n-1$, which yields
\begin{equation}\label{uddz}
u_1<u_2\le\cdots\le u_{n-1}<u_{n}.
\end{equation}
Noticing $T_0=S_0=T_n=S_n=0$ and substituting $i=1$ and $j=n$ into \eqref{gf2}, we obtain
\begin{equation}\label{xmyy}
\ds\f{d_2-d_1}{d_2}\sum_{k=1}^{n-1}\rho^{(2)}_{k+1}T_kg_k=0,
\end{equation}
which implies that $\{g_k\}_{k=1}^{n-1}$ must change sign. It follows from Lemma \ref{proffg} that $\bar i=\min\{i:g_i>0\}$ is well defined with $1<\bar i<n-1$. Moreover, $g_k\le0$ for $1\le k\le \bar i-1$ and $g_{\bar i}> 0$. This, combined with \eqref{figi-g} and \eqref{g-k}, implies that
\begin{equation}\label{yfh}
v_{\bar i+1}>v_{\bar i}\;\;\text{and}\;\;r-u_{\bar i}-v_{\bar i}<0.
\end{equation}
This, together with \eqref{uddz}, implies that $r-u_{i_g+1}-v_{i_g+1}<0$. Then by \eqref{g-k} and induction, we can prove that $g_k>0$ for $k=i_{g+1},\cdots,n-1$,
which contradicts Lemma \ref{proffg}. Therefore, $\{T_k\}_{k=1}^{n-1}$ must change sign. Similarly, we can prove that $\{S_k\}_{k=1}^{n-1}$ also changes sign.

Now we consider (ii). The proof of ${\rm (ii_1)}$ is similar to ${\rm (i_1)}$, so we omit it here.
To see ${\rm (ii_2)}$, suppose to the contrary that $\{T_k\}_{k=1}^{n-1}$ does not change sign. This, together with ${\rm (ii_1)}$, implies that $T_k\le0$ for $k=1,\cdots,n-1$, which yields $f_k<0$ for $k=1,\cdots,n-1$.
Then substituting $i=1$ and $j=n$ into \eqref{fg2}, we obtain
\begin{equation*}
\ds\f{d_1-d_2}{d_1}\sum_{k=1}^{n-1}\rho^{(1)}_{k+1}S_kf_k=0,
\end{equation*}
which means that $\{S_k\}_{k=1}^{n-1}$ must change sign. 
Hence by ${\rm (ii_1)}$,  
$$
\hat i=\max\{i:S_i>0\} \ \ \text{and}\ \ \check i=\min\{i:S_i>0\}
$$
are well defined with $1<\check i,\hat i<n-1$.
We first suppose $d_2>d_1$.
Noticing $S_k\le0$ for $\hat i+1< k<n$ and $S_{\hat i}>0$ and substituting $i=\hat i+1$ and $j=n$ into \eqref{fg2}, we obtain
\begin{equation*}
\begin{split}
0\ge&\ds\f{d_1-d_2}{d_1}\sum_{k=\hat i+1}^{n-1}\rho^{(1)}_{k+1}S_k f_k\\
=&-\rho^{(1)}_i\left[(d_1+q_{\hat i})v_{\hat i+1}T_{\hat i}-(d_2+q_{\hat i})u_{\hat i+1}S_{\hat i}\right]>0,
\end{split}
\end{equation*}
which is a contradiction. If $d_2<d_1$, we can obtain a contradiction by substituting $i=1$ and $j=\check i$ into \eqref{fg2}.
This shows that $\{T_k\}_{k=1}^{n-1}$ must change sign.  Similarly, we can prove that $\{S_k\}_{k=1}^{n-1}$ also changes sign.
\end{proof}

We are ready to show that there is no positive equilibrium when $n=2,3$.
\begin{lemma}\label{23no}
Suppose that $d_1, d_2, r>0$ with $d_1\ne d_2$, $\bm q\ge \bm 0$, and $(\rm{\bf{H1}})$ or $(\rm{\bf{H2}})$ holds. If $n=2,3$, then model \eqref{pat-cp2} has no positive equilibrium.
\end{lemma}
\begin{proof}
We only consider the case that $(\textbf{H1})$ holds, since the  case $(\textbf{H2})$  can be proved similarly. Suppose to the contrary that there exists a positive equilibrium $(\bm u,\bm v)$. By Lemma \ref{proffg}, we see that $\{g_k\}_1^{n-1}=\{g_1\}$ with $g_1<0$ if $n=2$, and $\{g_k\}_1^{n-1}=\{g_1,g_2\}$ with $g_1,g_2<0$ if $n=3$. It follows from Lemma \ref{tsbh} that $\{T_k\}_1^{n-1}=\{T_1\}$ with $T_1>0$ if $n=2$, and $\{T_k\}_1^{n-1}=\{T_1,T_2\}$ with $T_1,T_2>0$ if $n=3$. Then we have
\begin{equation*}
\f{d_2-d_1}{d_2}\sum_{k=1}^{n-1}\rho^{(2)}_{k+1}T_kg_k\ne0,
\end{equation*}
which contradicts \eqref{xmyy}. This completes the proof.
\end{proof}

Then we show the nonexistence of positive equilibrium when $n\ge 4$.
\begin{lemma}\label{iv2}
Suppose that $d_1, d_2, r>0$ with $d_1\ne d_2$, $\bm q\ge \bm 0$, and $(\rm{\bf{H1}})$ holds. If $n\ge 4$, then model \eqref{pat-cp2} has no positive equilibrium.
\end{lemma}
\begin{proof}
Since the nonlinear terms of model \eqref{pat-cp2} are symmetric, we only need to consider the case $d_1<d_2$. Suppose to the contrary that model \eqref{pat-cp2} admits a positive equilibrium $(\bm u,\bm v)$. It follows from Lemma \ref{tsbh} that there exists $2p\;(p\ge1)$ points $1\le i^{2p}<i^{2p-1}<\cdots<i^2<i^1<n-1$ given by
\begin{equation}\label{ti}
\begin{split}
&i^1=\max\{i:T_i<0\},\;\;i^2=\max\{i<i^1:T_i>0\},\cdots,\\
&i^{2p-1}=\max\{i<i^{2p-2}:T_i<0\},\;\;i^{2p}=\max\{i<i^{2p-1}:T_i>0\}.
\end{split}
\end{equation}
Similarly, $i_s=\max\{i:S_i<0\}$ is also well defined with $1<i_s<n-1$.

Suppose that $p=1$. Then we will obtain a contradiction for each of the following three cases:
\begin{equation*}
(A_1): i_s>i^1,\;\; (A_2):  i^2\le i_s\le i^1,\;\;(A_3): i_s< i^2.
\end{equation*}
First, we consider case $(A_1)$. By the definition of $i^1$ and $i_s$,  we have
\begin{equation}\label{sis1}
S_{i_s}<0,\;\;S_k\ge0\;\;\text{for}\;\;k=i_s+1,\cdots,n-1,\;\;\text{and}\;\;T_{k}\ge0\;\;\text{for}\;\;k=i_s,\cdots,n-1.
\end{equation}
It follows that 
\begin{equation}\label{sis3}
v_{i_s}>v_{i_s+1}, \ v_{i_s+1}\le\cdots\le v_n\;\;\text{and}\;\;u_{i_s}\le u_{i_s+1}\le\cdots\le u_n.
\end{equation}
Then $g_{i_s}=d_2v_{i_s+1}-(d_2+q_k)v_{i_s}<0$. We claim that
\begin{equation}\label{gddx}
g_k\le0\;\;\text{for}\;\;k=i_s+1,\cdots,n-1.
\end{equation}
If it is not true, then $i_g=\min\{i>i_s:g_i>0\}$ is well-defined, and consequently,
\begin{equation*}
g_k\le0\;\;\text{for}\;\;k=i_s,\cdots,i_g-1\;\;\text{and}\;\;g_{i_g}>0.
\end{equation*}
This, combined with \eqref{g-k-a}, implies that
\begin{equation}\label{initial}
r-u_{i_g}-v_{i_g}<0.
\end{equation}
Then, by \eqref{g-k-a}, \eqref{sis3} and \eqref{initial}, we have
\begin{equation*}
g_k-g_{k-1}=-v_k(r-u_k-v_k)>0\;\;\text{for}\;\;k=i_g+1,\cdots,n-1.
\end{equation*}
Noticing that $g_{i_g}>0$, we have
\begin{equation*}
0<g_{i_g}<\cdots<g_{n-1},
\end{equation*}
which contradicts Lemma \ref{proffg}. This proves \eqref{gddx}. Then substituting $i=i_s+1$ and $j=n$ into \eqref{gf2} and noticing  $T_n=S_n=0$, we see from \eqref{sis1} and \eqref{gddx} that
\begin{equation*}
\begin{split}
0&<-\rho^{(2)}_{i_s+1}\left[(d_2+q_{i_s})u_{i_s+1}S_{i_s}-(d_1+q_{i_s})v_{i_s+1}T_{i_s}\right]\\
&=\ds\f{d_2-d_1}{d_2}\sum_{k=i_s+1}^{n-1}\rho^{(2)}_{k+1}T_kg_k\le 0,
\end{split}
\end{equation*}
which is a contradiction.

Then we consider $(A_2)$. By the definition of $i_s$, we have $S_{i_s}<0$ and $S_{i_s+1}\ge0$. By \eqref{TSk-b}, we see that
\begin{equation}\label{sis2}
r+q_{i_s}-q_{i_s+1}-u_{i_s+1}-v_{i_s+1}<0.
\end{equation}
By the definition of $i^2$, we have  $T_{i^2}>0$ and $T_{i^2+1}\le0$. If $i_s=i^2$, then by \eqref{TSk-a}, we have
\begin{equation*}
0>d_1T_{i_s+1}-(d_1+q_{i_s})T_{i_s}=-u_{i_s+1}(r+q_{i_s}-q_{i_s+1}-u_{i_s+1}-v_{i_s+1})>0,
\end{equation*}
which is a contradiction.

If $i_s>i^2$, then $T_{k}\le0$ for $k=i^2+1,\cdots,i_s$, and consequently,
\begin{equation*}
u_{i^2+1}\ge\cdots\ge u_{i_s} \ge u_{i_s+1}.
\end{equation*}
This, combined with $S_{i_s}<0$, assumption $(\textbf{H1})$ and \eqref{sis2}, implies that
\begin{equation}
r+q_{i_s-1}-q_{i_s}-u_{i_s}-v_{i_s}<0.
\end{equation}
Then, by \eqref{TSk-b}, we have $S_{i_s-1}<0$, which implies that $v_{i_s-1}>v_{i_s}$. By induction, we can show that
\begin{equation*}
r+q_{k-1}-q_{k}-u_{k}-v_{k}<0\;\;\text{for}\;\;k=i^2+1,\cdots,i_s,
\end{equation*}
and
\begin{equation*}
S_k<0\;\;\text{for}\;\;k=i^2,\cdots,i_s.
\end{equation*}
By \eqref{TSk-a} with $k=i^2+1$, we have
\begin{equation*}
0>d_1T_{i^2+1}-(d_1+q_{i^2})T_{i^2}=-u_{i^2+1}(r+q_{i^2}-q_{i^2+1}-u_{i^2+1}-v_{i^2+1})>0,
\end{equation*}
which is a contradiction.

For case $(A_3)$, we have
\begin{equation}\label{a31}
T_k\ge0\;\;\text{for}\;\;k=i_s+1,\cdots,i^2\;\;\text{and}\;\;S_k\ge0\;\;\text{for}\;\;k=i_s+1,\cdots,n-1.
\end{equation}
It follows that
\begin{equation}\label{is}
u_{i_s+1}\le\cdots\le u_{i^2+1}\;\;\text{and}\;\;v_{i_s+1}\le\cdots\le v_{n}.
\end{equation}
By the definition of $i_s$ and $i^2$, we have $S_{i_s}<0$ and $T_{i^2+1}\le0$, which yields $g_{i_s}<0$ and $f_{i^2+1}<0$. We claim that
\begin{equation}\label{gis}
f_k\le0\;\;\text{for}\;\;k=i_s+1,\cdots,i^2.
\end{equation}
If it is not true, then $i_f=\max\{i^s+1\le i\le i^2:f_i>0\}$ is well-defined, and consequently,
\begin{equation}\label{fif}
f_{i_f}>0\;\;\text{and}\;\;f_k\le0\;\;\text{for}\;\;k=i_f+1,\cdots,i^2+1.
\end{equation}
By $f_{i_f}>0$, $f_{i_f+1}\le 0$ and \eqref{f-k-a}, we have $r-u_{i_f+1}-v_{i_f+1}>0$. Then by \eqref{is}, we see that
\begin{equation*}
r-u_k-v_k>0\;\;\text{for}\;\;k=i_s+1,\cdots,i_f.
\end{equation*}
This, combined with \eqref{f-k-a}, \eqref{g-k-a}, yields
\begin{equation}\label{if2}
f_{i_s}>\cdots>f_{i_f}>0\;\;\text{and}\;\;0>g_{i_s}>\cdots>g_{i_f}.
\end{equation}
where we have used $f_{i_f}>0$ and $g_{i_s}<0$.

If $i_f<i^2$, then substituting $i=i_f+1$ and $j=i^2+1$ into \eqref{fg2}, we see from \eqref{a31}, \eqref{fif} and \eqref{if2} that
\begin{equation*}
\begin{split}
0\le &\ds\f{d_1-d_2}{d_1}\sum_{k=i_f+1}^{i^2}\rho^{(1)}_{k+1}S_kf_k\\
=&\rho^{(1)}_{i^2+1}\left(d_1v_{i^2+1}T_{i^2+1}-d_2u_{i^2+1}S_{i^2+1}\right)-\rho^{(1)}_{i_f+1}\left[(d_1+q_{i_f})v_{i_f+1}T_{i_f}-(d_2+q_{i_f})u_{i_f+1}S_{i_f}\right]\\
=&\rho^{(1)}_{i^2+1}\left(d_1v_{i^2+1}T_{i^2+1}-d_2u_{i^2+1}S_{i^2+1}\right)-\rho^{(1)}_{i_f+1}\left(v_{i_f+1}f_{i_f}-u_{i_f+1}g_{i_f}\right)<0,\\
\end{split}
\end{equation*}
which is a contradiction.

If $i_f=i^2$, then substituting $i=i_s+1$ and $j=i^2+1$ into \eqref{gf2}, and from \eqref{a31} and \eqref{if2}, we have a contradiction:
\begin{equation*}
\begin{split}
0>&\ds\f{d_2-d_1}{d_2}\sum_{k=i_s+1}^{i^2}T_kg_k\rho^{(2)}_{k+1}\\
=&\rho^{(2)}_{i^2+1}\left(d_2u_{i^2+1}S_{i^2+1}-d_1v_{i^2+1}T_{i^2+1}\right)-\rho^{(2)}_{i_s+1}\left[(d_2+q_{i_s})u_{i_s+1}S_{i_s}-(d_1+q_{i_s})v_{i_s+1}T_{i_s}\right]\\
=&\rho^{(2)}_{i^2+1}\left(d_2u_{i^2+1}S_{i^2+1}-d_1v_{i^2+1}T_{i^2+1}\right)-\rho^{(2)}_{i_s+1}\left(u_{i_s+1}g_{i_s}-v_{i_s+1}f_{i_s}\right)>0.\\
\end{split}
\end{equation*}
This proves \eqref{gis}.

Now we will obtain a contradiction for case $(A_3)$. Recall that $T_{i_s}\ge0$, $T_{i^2+1}\le 0$,  $S_{i_s}<0$ and $S_{k}\ge0$ for any $k\ge i_{s}+1$. Then substituting $i=i_s+1$ and $j=i^2+1$ into \eqref{fg2}, we see from \eqref{gis} that
\begin{equation*}
\begin{split}
0\le &\ds\f{d_1-d_2}{d_1}\sum_{k=i_s+1}^{i^2}\rho^{(1)}_{k+1}S_kf_k \\
=&\rho^{(1)}_{i^2+1}\left(d_1v_{i^2+1}T_{i^2+1}-d_2u_{i^2+1}S_{i^2+1}\right)\\
&-\rho^{(1)}_{i_s+1}\left[(d_1+q_{i_s})v_{i_s+1}T_{i_s}-(d_2+q_{i_s})u_{i_s+1}S_{i_s}\right]<0,
\end{split}
\end{equation*}
which is a contradiction.

Suppose $p=2$. Then we need to consider  $(A_1)-(A_2)$ and the following three cases:
\begin{equation*}
(B_1):  i^3<{i}_s< i^2,\;\; (B_2):  i^4\le {i}_s\le i^3,\;\;(B_3): {i}_s< i^4.
\end{equation*}
For each case of $(A_1)$-$(A_2)$ and $(B_1)$, we can obtain a contradiction by repeating the proof of $p=1$. Cases $(B_2)$-$(B_3)$ can be handled exactly the same as $(A_2)$-$(A_3)$ by replacing $i^2$ by $i^4$ and $i^1$ by $i^3$, respectively. Continue this process, we can prove the desired result for any $p>1$.   
\end{proof}

\begin{lemma}\label{iv21}
Suppose that $d_1, d_2, r>0$ with $d_1\ne d_2$, $\bm q\ge \bm 0$, and  $(\rm{\bf{H2}})$ holds. If $n\ge 4$, then model \eqref{pat-cp2} has no positive equilibrium.
\end{lemma}
\begin{proof}
Similar to the proof of Lemma \ref{iv2}, we may assume $d_1<d_2$ and suppose to the contrary that the model has a positive equilibrium $(\bm u, \bm v)$.  It follows from Lemma \ref{tsbh} that there exists $2p\;(p\ge1)$ points $1\le i_{2p}<i_{2p-1}<\cdots<i_2<i_1<n-1$ given by
\begin{equation}\label{ti2}
\begin{split}
&i_1=\max\{i:T_i>0\},\;\;i_2=\max\{i<i_1:T_i<0\},\cdots,\\
&i_{2p-1}=\max\{i<i_{2p-2}:T_i>0\},\;\;i_{2p}=\max\{i<i_{2p-1}:T_i<0\} 
\end{split}
\end{equation}
and $i^s=\max\{i:S_i>0\}$ is well-defined.

Suppose that $p=1$. Then we consider the following three cases:
\begin{equation*}
(A_1): i^s>i_1,\;\; (A_2):  i_2\le i^s\le i_1,\;\;(A_3): i^s< i_2.
\end{equation*}
For case $(A_1)$, we have
\begin{equation*}
S_{i^s}>0,\;\;S_k\le0\;\;\text{for}\;\;k=i^s+1,\cdots,n-1,\;\;\text{and}\;\;T_{k}\le0\;\;\text{for}\;\;k=i^s,\cdots,n-1,
\end{equation*}
which yields that $f_{k}<0$ for $k=i^s,\cdots,n-1$.
Then substituting $i=i^s+1$ and $j=n$ into \eqref{fg2}, we have
\begin{equation*}
\begin{split}
0\ge&\ds\f{d_1-d_2}{d_1}\sum_{k=i^s+1}^{n-1}\rho^{(1)}_{k+1}S_k f_k=-\rho^{(1)}_{i^s+1}\left[(d_1+q_{i^s})v_{i^s+1}T_{i^s}-(d_2+q_{i^s})u_{i^s+1}S_{i^s}\right]>0,
\end{split}
\end{equation*}
which is a contradiction.

For case $(C_2)$, we have $S_{i^s}>0$, $S_{i^s+1}\le0$, $T_{i_2}<0$ and $T_{i_2+1}\ge0$. If $i^s=i_2$, then by \eqref{TSk}, we have
\begin{equation}\label{c1}
0<\frac{1}{u_{i^s+1}}[d_1T_{i^s+1}-(d_1+q_{i^s})T_{i^s}]
=\frac{1}{v_{i^s+1}}[d_2S_{i^s+1}-(d_2+q_{i^s})S_{i^s}]<0,
\end{equation}
which is a contradiction.

If $i^s>i_2$, then $T_{k}\ge0$ for $k=i_2+1,\cdots,i^s$, and consequently,
\begin{equation}\label{a2}
u_{i_2+1}\le\cdots\le u_{i^s}\le u_{i^s+1}.
\end{equation}
By \eqref{TSk} again, and noticing that $S_{i^s}>0$ and $S_{i^s+1}\le0$, we have
\begin{equation*}
r+q_{i^s}-q_{i^s+1}-u_{i^s+1}-v_{i^s+1}>0.
\end{equation*}
This, combined with $S_{i^s}>0$, assumption $(\textbf{H2})$ and \eqref{a2}, implies that
\begin{equation}
r+q_{i^s-1}-q_{i^s}-u_{i^s}-v_{i^s}>0.
\end{equation}
Then by \eqref{TSk-b} and induction, we can show that
\begin{equation*}
r+q_{k-1}-q_{k}-u_{k}-v_{k}>0\;\;\text{for}\;\;k=i_2+1,\cdots,i^s,  
\end{equation*}
and
\begin{equation*}
S_k>0\;\;\text{for}\;\;k=i_2,\cdots,i^s.
\end{equation*}
By \eqref{TSk-a} with $k=i_2+1$, we have that
\begin{equation*}
0<d_1T_{i_2+1}-(d_1+q_{i_2})T_{i_2}=-u_{i_2+1}(r+q_{i_2}-q_{i_2+1}-u_{i_2+1}-v_{i_2+1})<0,
\end{equation*}
which is a contradiction.

For case $(C_3)$, we have $S_{i^s}>0$, $S_k\le0$ for $k\ge i^s+1$, $T_{i_2+1}\ge0$, and $T_k\le0$ for $k=i^s,\cdots,i_2$, which implies that $f_k<0$ for $k=i^s,\cdots,i_2$. Then substituting $i=i^s+1$ and $j=i_2+1$ into \eqref{fg2}, we see that
\begin{equation*}
\begin{split}
0\ge&\ds\f{d_1-d_2}{d_1}\sum_{k=i^s+1}^{i_2}\rho^{(1)}_{k+1}S_kf_k\\
=&\rho^{(1)}_{i_2+1}\left(d_1v_{i_2+1}T_{i_2+1}-d_2u_{i_2+1}S_{i_2+1}\right)\\
&-\rho^{(1)}_{i^s+1}\left[(d_1+q_{i^s})v_{i^s+1}T_{i^s}-(d_2+q_{i^s})u_{i^s+1}S_{i^s}\right]>0,
\end{split}
\end{equation*}
which is a contradiction.

Similar to the proof of Lemma \ref{iv2}, we can inductively prove the desired result for any $p>1$. 
\end{proof}

\section{Local stability of semi-trivial equilibria}
In this section, we consider the stability of the two semi-trivial equilibria $(\bm u^*,\bm 0)$  and $(\bm 0, \bm v^*)$ if they exist. It suffices to  study only
the stability of $(\bm u^*,\bm 0)$ since the stability of $(\bm 0, \bm v^*)$ can be investigated similarly.  It is easy to see that the stability of $(\bm u^*, \bm 0)$, if exists, is determined by the sign of $\lambda_1(d_2, \bm q, \bm r-\bm u^*)$: if  $\lambda_1(d_2, \bm q, \bm r-\bm u^*)<0$, then $(\bm u^*, \bm 0)$ is locally asymptotically stable;  if  $\lambda_1(d_2, \bm q, \bm r-\bm u^*)>0$, then $(\bm u^*, \bm 0)$ is unstable. Similarly, the stability of $(\bm 0, \bm v^*)$, if exists, is determined by the sign of $\lambda_1(d_1, \bm q, \bm r-\bm v^*)$.

\begin{lemma}\label{wdx}
Let $d_1, d_2, r>0$ and $\bm q\ge \bm 0$. Suppose that the semi-trivial equilibrium $(\bm u^*,\bm 0)$ of model \eqref{pat-cp2} exists. Then
the following statements hold:
\begin{enumerate}
    \item[$\rm{(i)}$] If $(\rm{\bf{H1}})$ holds,  then $(\bm u^*,\bm 0)$ is locally asymptotically stable when $d_1>d_2$ and
     unstable when $d_1<d_2$;
    \item[$\rm{(ii)}$] If $(\rm{\bf{H2}})$ holds, then $(\bm u^*,\bm 0)$ is locally asymptotically stable when  $d_1<d_2$ and unstable when
    $d_1>d_2$.
\end{enumerate}
\end{lemma}
\begin{proof}
Let $\bm\phi=(\phi_1,\cdots,\phi_n)^T\gg0$ be an eigenvector corresponding to $\lambda_1(d_2,\bm q,\bm r-\bm u^*)$, where $\bm r=(r, \dots, r)$.

(i) It follows from Lemma \ref{wdx0} that
\begin{equation}\label{ustar}
u^*_i<u^*_{i+1}\;\;\text{for}\;\;i=1,\cdots,n-1.
\end{equation}
This, combined with Lemma \ref{prop-p11}, implies that
\begin{equation}\label{phiin}
d_2\phi_{i+1}-(d_2+q_i)\phi_i<0\;\;\text{for}\;\;i=1,\cdots,n-1.
\end{equation}
Since $\bm u^*$ is an eigenvector corresponding to $\lambda_1(d_1,\bm q,\bm r-\bm u^*)=0$, by Lemma \ref{prpla}, we have
\begin{equation}\label{estimla}
\begin{split}
&-\la_1(d_2,\bm q,\bm r-\bm u^*)\sum_{k=1}^n\rho^{(2)}_ku^*_k\phi_k\\
=&\f{(d_2-d_1)}{d_2}\sum_{k=1}^{n-1}\rho^{(2)}_{k+1}
\left(u^*_{k+1}-u^*_k\right)\left[d_2\phi_{k+1}-(d_2+q_k)\phi_k\right],
\end{split}
\end{equation}
where $\rho^{(2)}_i$, $i=1,\cdots,n$, are defined by \eqref{rho2}.
Hence by \eqref{ustar}-\eqref{phiin},  $\la_1(d_2,\bm q,\bm r-\bm u^*)<0$ when $d_1>d_2$ and $\la_1(d_2,\bm q,\bm r-\bm u^*)>0$ when $d_1<d_2$. Therefore, $(\bm u^*,\bm 0)$ is locally asymptotically stable when $d_1>d_2$ and unstable when $d_1<d_2$.

(ii) We first claim that
\begin{equation}\label{d1d2l}
\la_1(d_2,\bm q,\bm r-\bm u^*)\neq0\;\;\text{if}\;\;d_2\ne d_1.
\end{equation}
Without loss of generality, we may assume  $d_1<d_2$.
Suppose to the contrary that $\la_1(d_2,\bm q,\bm r-\bm u^*)=0$. Let $\bm \psi\gg0$ be a corresponding eigenvector of $\la_1(d_2,\bm q,\bm r-\bm u^*)$. Recall that $\bm u^*$ is an eigenvector corresponding to $\lambda_1(d_1,\bm q,\bm r-\bm u^*)\;(=0)$. We define two sequences $\{\widetilde f_k\}_{k=0}^n$ and $\{\widetilde g_k\}_{k=0}^n$ as follows:
\begin{equation}\label{tildefg}
\begin{split}
&\widetilde f_0=\widetilde f_n=0,\;\;\widetilde f_k=d_1u^*_{k+1}-(d_1+q_k)u^*_k\;\;\text{for}\;\;k=1,\cdots,n-1,\\
&\widetilde g_0=\widetilde g_n=0,\;\;\widetilde g_k=d_2\psi_{k+1}-(d_2+q_k)\psi_k\;\;\text{for}\;\;k=1,\cdots,n-1,
\end{split}
\end{equation}
and another two sequences $\{\widetilde T_k\}_{k=0}^n$ and $\{\widetilde S_k\}_{k=0}^n$ as follows:
\begin{equation}\label{tildets}
\begin{split}
&\widetilde T_0=\widetilde T_n=0,\;\;\widetilde T_k=u^*_{k+1}-u^*_k\;\;\text{for}\;\;k=1,\cdots,n-1,\\
&\widetilde S_0=\widetilde S_n=0,\;\;\widetilde S_k=\psi_{k+1}-\psi_k\;\;\text{for}\;\;k=1,\cdots,n-1.
\end{split}
\end{equation}

Using  similar arguments as the proof of Lemma \ref{indfg}, we can show that
\begin{equation}\label{tildef}
\begin{split}
\ds\f{d_1-d_2}{d_1}\sum_{k=i}^{j-1}\rho^{(1)}_{k+1}\widetilde S_k\widetilde f_k=&\rho^{(1)}_j\left(d_1\psi_j\widetilde T_j
-d_2u^*_j\widetilde S_j\right)\\
&-\rho^{(1)}_i\left[(d_1+q_{i-1})\psi_i\widetilde T_{i-1}-(d_2+q_{i-1})u^*_i\widetilde S_{i-1}\right],
\end{split}
\end{equation}
for any $1\le i<j\le n$, where $\rho^{(1)}_i$, $i=1,\cdots,n$, are defined by \eqref{rho1}.

By Lemma \ref{wdx0}, we have
\begin{equation}\label{tildetf}
\widetilde T_k<0\;\;\text{and}\;\;\widetilde f_k<0\;\;\text{for}\;\;k=1,\cdots,n-1.
\end{equation}
It follows that
\begin{equation}\label{tildesf}
\widetilde S_{n-1}=\frac{\psi_n}{d_2+q_{n-1}}(r+q_{n-1}-q_n-u^*_n)=\frac{\psi_n}{d_2+q_{n-1}}\f{(d_1+q_{n-1})\widetilde T_{n-1}}{u^*_n}<0.
\end{equation}
Substituting $i=1$ and $j=n$ into \eqref{tildef}, we have
\begin{equation}\label{tilde1n}
\sum_{k=1}^{n-1}\rho^{(1)}_{k+1}\widetilde S_k \widetilde f_k=0.
\end{equation}
Then it follows from \eqref{tildesf}-\eqref{tildetf} that  $i^*=\max\{i:\widetilde S_{i}>0\}$ is well defined with $1\le i^*<n-1$. By the definition of $i^*$, we have
\begin{equation*}
\widetilde S_{i^*}>0\;\;\text{and}\;\;\widetilde S_k\le0\;\;\text{for}\;\;k=i^*+1,\cdots,n-1.
\end{equation*}
Substituting $i=i^*+1$ and $j=n$ into \eqref{tildef}, we have
\begin{equation*}
\begin{split}
0\ge&\ds\f{d_1-d_2}{d_1}\sum_{k=i^*+1}^{n-1}\rho^{(1)}_{k+1}\widetilde S_k \widetilde f_k\\
=&-\rho^{(1)}_i\left[(d_1+q_{i^*})\psi_{i^*+1}\widetilde T_{i^*}-(d_2+q_{i^*})u^*_{i^*+1}\widetilde S_{i^*}\right]>0,
\end{split}
\end{equation*}
which is a contradiction.
This proves that $\la_1(d_2,\bm q,\bm r-\bm u^*)\neq0$ if $d_1\neq d_2$. 

Finally, it follows from \eqref{fzjs} that
\begin{equation*}
\left.\f{\partial}{\partial d}\la_1(d,\bm q,\bm r-\bm u^*)\right|_{d=d_1}
=\f{\sum_{i=1}^{n-1}\rho^{(1)}_{i+1}(u^*_{i+1}-u^*_i)\left[(d_1+q_i)u^*_i-d_1u^*_{i+1}\right]}
{\sum_{i=1}^{n}d_1\rho^{(1)}_i(u^*_i)^2}<0.
\end{equation*}
This implies that $\la_1(d_2,\bm q,\bm r-\bm u^*)<0$ when $0<d_2-d_1\ll 1$ and $\la_1(d_2,\bm q,\bm r-\bm u^*)>0$ when $0<d_1-d_2\ll 1$. Hence by \eqref{d1d2l}, we know that $\la_1(d_2,\bm q,\bm r-\bm u^*)<0$ when $d_1<d_2$ and $\la_1(d_2,\bm q,\bm r-\bm u^*)>0$ when $d_1>d_2$. Therefore, $(\bm u^*,\bm 0)$ is locally asymptotically stable when $d_1<d_2$ and unstable when $d_1>d_2$.
\end{proof}

\bibliographystyle{plain}
\bibliography{BibforControl}

\end{document}